\documentclass[10pt,a4paper, conference, onecolumn]{IEEEtran}
\usepackage[utf8]{inputenc}
\usepackage{graphicx}
\usepackage{amssymb}
\usepackage[cmex10]{amsmath}
\usepackage{amsthm}
\usepackage{array}
\usepackage[usenames,dvipsnames]{xcolor}
\usepackage{hyperref}
\hypersetup{colorlinks=false, pdfborderstyle={/S/U/W 1},pdfborder=0 0 1, citebordercolor=Blue, filebordercolor=Red, linkbordercolor=Red, urlbordercolor=Blue}
\usepackage[english]{babel}
\usepackage[english = american]{csquotes} 
\usepackage{url}
\usepackage{xparse}
\usepackage{subfigure}
\usepackage[backend=bibtex8, style=ieee, url=false, isbn=false, doi=true, maxnames=5, firstinits=true]{biblatex}
\AtEveryBibitem{
  \clearfield{month}
  \clearfield{edition}
  \clearname{editor}
}
\MakeOuterQuote{"}
\renewbibmacro{in:}{}
\DeclareListFormat{language}{}
\bibliography{QCCex}

\usepackage[ruled,vlined,titlenumbered]{algorithm2e}
\usepackage{tabularx}
\makeatletter
\newcommand{\removelatexerror}{\let\@latex@error\@gobble}
\makeatother
\newcommand{\printalgo}[1]
{
\begin{center}
\scalebox{0.97}{
\removelatexerror
\begin{algorithm}[H]
 #1
\end{algorithm}
}
\end{center}
}

\usepackage{pgf}
\usepackage{tikz}
\usetikzlibrary{matrix,chains,positioning,arrows,calc,decorations,plotmarks,patterns, fit,backgrounds}
\usepackage{pgfplots}
\usetikzlibrary{external}                       
\tikzsetexternalprefix{tikz/}                   
\tikzset{external/system call={pdflatex \tikzexternalcheckshellescape -interaction=batchmode -jobname "\image" "\texsource"}}
\pgfplotsset{compat=newest}
\tikzstyle{help lines}=[black!20,dashed]

\pgfplotscreateplotcyclelist{mylist}{red,blue,black,yellow,brown}
\newtheorem{definition}{Definition}
\newtheorem{theorem}[definition]{Theorem}
\newtheorem{lemma}[definition]{Lemma}

\newtheorem{example}[definition]{Example}
\newtheorem{corollary}[definition]{Corollary}

\newtheorem{conjecture}[definition]{Conjecture}
\DeclareMathOperator{\lcm}{lcm}
\DeclareMathOperator{\diag}{diag}

\DeclareMathOperator{\defi}{def}
\newcommand{\defeq}{\overset{\defi}{=}}
\newcommand{\defeqspace}{\overset{\hphantom{\defi}}{=}}
\newcommand{\defequiv}{\overset{\defi}{\equiv}}
\newcommand{\F}[1]{\mathbb{F}_{#1}}
\newcommand{\Fx}[1]{\ensuremath{\mathbb{F}_{#1}[X]}}

\renewcommand{\vec}[1]{\mathbf #1}
\newcommand{\shiftset}[2]{\ensuremath{{#1}^{\oplus{#2}}}}
\newcommand{\LIN}[4]{\ensuremath{[#1,#2,#3]_{#4}}}
\newcommand{\level}{\ensuremath{r}}
\newcommand{\basis}{\ensuremath{\tilde{K}}}
\newcommand{\QCC}{\ensuremath{\mathcal{C}}}
\newcommand{\QCCn}{\ensuremath{n}} 
\newcommand{\QCCm}{\ensuremath{m}} 
\newcommand{\QCCl}{\ensuremath{\ell}} 
\newcommand{\QCCk}{\ensuremath{k}} 
\newcommand{\QCCd}{\ensuremath{d}} 
\newcommand{\QCCs}{\ensuremath{s}} 
\newcommand{\QCCa}{\ensuremath{\mathcal{A}}}
\newcommand{\QCCan}{\ensuremath{n_{A}}} 
\newcommand{\QCCam}{\ensuremath{m_{A}}} 
\newcommand{\QCCal}{\ensuremath{\ell_{A}}}
\newcommand{\QCCak}{\ensuremath{k_{A}}} 
\newcommand{\QCCad}{\ensuremath{d_{A}}} 
\newcommand{\QCCas}{\ensuremath{s_{A}}} 
\newcommand{\QCCb}{\ensuremath{\mathcal{B}}}
\newcommand{\QCCbn}{\ensuremath{n_{B}}} 
\newcommand{\QCCbm}{\ensuremath{m_{B}}} 
\newcommand{\QCCbl}{\ensuremath{\ell_{B}}}
\newcommand{\QCCbk}{\ensuremath{k_{B}}} 
\newcommand{\QCCbd}{\ensuremath{d_{B}}} 
\newcommand{\QCCbs}{\ensuremath{s_{B}}} 
\newcommand{\CYCb}{\ensuremath{\mathcal{B}}}
\newcommand{\CYCbn}{\ensuremath{n_{B}}} 
\newcommand{\CYCbk}{\ensuremath{k_{B}}} 
\newcommand{\CYCbd}{\ensuremath{d_{B}}} 

\newcommand{\inta}{\ensuremath{a}}
\newcommand{\intb}{\ensuremath{b}}
\newcommand{\HTconst}{\ensuremath{f}}
\newcommand{\HTconsta}{\ensuremath{f_1}}
\newcommand{\HTconstb}{\ensuremath{f_2}}
\newcommand{\lenseq}{\ensuremath{\delta}}
\newcommand{\noseq}{\ensuremath{\nu}}
\newcommand{\HTmult}{\ensuremath{z}}
\newcommand{\HTmulta}{\ensuremath{z_1}}
\newcommand{\HTmultb}{\ensuremath{z_2}}

\newcommand{\eigenvalue}[1]{\ensuremath{\lambda_{#1}}}
\newcommand{\map}[2]{\ensuremath{\mu(#1,#2)}}
\newcommand{\mapb}[2]{\ensuremath{\bar{\mu}(#1,#2)}}
\newcommand{\submap}[2]{\ensuremath{\nu(#1,#2)}}

\DeclareDocumentCommand \eigenvector { ooo }
{
	\IfNoValueTF {#3}
	{  
		\IfNoValueTF {#2}
			{ 
				\IfNoValueTF {#1}
				{ 	
					\ensuremath{\mathbf{v}}	
				}
				{ 
					\ensuremath{v_{#1}}	
				}
			}
			{ 
			\ensuremath{\mathbf{v}_{#1}^{\langle #2 \rangle}}			
			}
	}
	{
	\ensuremath{{v}_{#1,#3}^{\langle #2 \rangle}}
	}
}
\DeclareDocumentCommand \gen { ooo }
{
	\IfNoValueTF {#3}
	{  
		\IfNoValueTF {#2}
			{ 
				\IfNoValueTF {#1}
				{ 	
					\ensuremath{g(X)}	
				}
				{ 
					\ensuremath{g^{#1}(X)}	
				}
			}
			{ 
			\ensuremath{g_{#1,#2}(X)}			
			}
	}
	{
	\ensuremath{g_{#2,#3}^{#1}(X)}
	}
}
\DeclareDocumentCommand \genunred { ooo }
{
	\IfNoValueTF {#3}
	{  
		\IfNoValueTF {#2}
			{ 
				\IfNoValueTF {#1}
				{ 	
					\ensuremath{u(X)}	
				}
				{ 
					\ensuremath{u^{#1}(X)}	
				}
			}
			{ 
              \ensuremath{u_{#1,#2}(X)}			
			}
	}
	{
      \ensuremath{u_{#2,#3}^{#1}(X)}
	}
}

\DeclareDocumentCommand \genbar { ooo }
{
	\IfNoValueTF {#3}
	{  
		\IfNoValueTF {#2}
			{ 
				\IfNoValueTF {#1}
				{ 	
					\ensuremath{\bar{g}(X)}	
				}
				{ 
					\ensuremath{\bar{g}^{#1}(X)}	
				}
			}
			{ 
              \ensuremath{\bar{g}_{#1,#2}(X)}			
			}
	}
	{
      \ensuremath{\bar{g}_{#2,#3}^{#1}(X)}
	}
}
\DeclareDocumentCommand \genarg { moom }
{
  \IfNoValueTF {#2}
  {
    \ensuremath{g^{#1}(#4)}   
  }
  { 
    \ensuremath{g^{#1}_{#2,#3}(#4)}
  }
}
\DeclareDocumentCommand \gencoeff{ mm }
{
  { 
    \ensuremath{g^{#1}_{#2}}
  }
}
\DeclareDocumentCommand \genmat { oo }
{
  \IfNoValueTF {#2}
  { 
    \IfNoValueTF {#1}
    {   
      \ensuremath{\mathbf{G}(X)}
    }
    { 
      \ensuremath{\mathbf{G}^{#1}(X)}    
    }
  }
  { 
    \ensuremath{\mathbf{G}^{#1}(#2)}           
  }
}

\DeclareDocumentCommand \supp{o}
{
	\IfNoValueTF {#1}
	{  
		\ensuremath{\mathcal{Y}}
	}
	{
		\ensuremath{\mathcal{Y}_{#1}}
	}
}
\DeclareDocumentCommand \errorsup{o}
{
	\IfNoValueTF {#1}
	{  
		\ensuremath{\mathcal{E}}
	}
	{
		\ensuremath{\mathcal{E}_{#1}}
	}
}
\DeclareDocumentCommand \noerrors{o}
{
	\IfNoValueTF {#1}
	{  
		\ensuremath{\mathcal{\varepsilon}}
	}
	{
		\ensuremath{\mathcal{\varepsilon}_{#1}}
	}
}
\DeclareDocumentCommand \eigenspace{o}
{
	\IfNoValueTF {#1}
	{  
		\ensuremath{\mathcal{V}}
	}
	{
		\ensuremath{\mathcal{V}_{#1}}
	}
}
\newcommand{\eigencode}[1]{\ensuremath{\mathbb{C}(#1)}}
\newcommand{\mult}[1]{\ensuremath{u_{#1}}}
\newcommand{\ECn}{\ensuremath{n^{ec}}}
\newcommand{\ECk}{\ensuremath{k^{ec}}}
\newcommand{\ECd}{\ensuremath{d^{ec}}}
\newcommand{\coset}[1]{\ensuremath{M_{#1}}}
\DeclareDocumentCommand \minpoly { mo }
{
  \IfNoValueTF {#2}
  {
    \ensuremath{M_{#1}(X)}   
  }
  { 
    \ensuremath{M_{#1}(#2)}
  }
}
\newcommand{\NewBCHBound}{\ensuremath{d^{\ast}}}

\newcommand{\rowop}[1]{\ensuremath{\mathsf{R}[#1]}}

\renewcommand{\tilde}{\widetilde}
\renewcommand{\bar}{\overline}
\newcommand{\interval}[1]{\ensuremath{[#1)}}
\newcommand\bigzero{\makebox(0,0){\text{\textnormal{\huge0}}}}

\newcommand{\ith}[1]{\ensuremath{#1^{\text{th}}}}
\newcommand{\nonzero}{nonzero}
\addto\captionsenglish{}
\begin{document}
\IEEEoverridecommandlockouts
\title{Spectral Analysis of Quasi-Cyclic Product Codes}
\author{\IEEEauthorblockN{Alexander Zeh}\thanks{A. Zeh has been supported by the German research council (Deutsche Forschungsgemeinschaft, DFG) under grant Ze1016/1-1. S. Ling has been supported by NTU Research Grant M4080456. Parts of the presented work were published in the proceedings of the $10^{\text{th}}$ International ITG Conference on Systems, Communications and Coding 2015 (SCC’2015)~\cite{zeh_construction_2015}.}
\IEEEauthorblockA{Computer Science Department\\
Technion---Israel Institute of Technology\\
Haifa, Israel\\
\texttt{alex@codingtheory.eu}
}
\and
\IEEEauthorblockN{San Ling}
\IEEEauthorblockA{Division of Mathematical Sciences, School of Physical \& \\
Mathematical Sciences, Nanyang Technological University\\
Singapore, Republic of Singapore\\
\texttt{lingsan@ntu.edu.sg}}
}

\maketitle

\begin{abstract}
This paper considers a linear quasi-cyclic product code of two given quasi-cyclic codes of relatively prime lengths over finite fields. We give the spectral analysis of a quasi-cyclic product code in terms of the spectral analysis of the row- and the column-code. Moreover, we provide a new lower bound on the minimum Hamming distance of a given quasi-cyclic code and present a new algebraic decoding algorithm.

More specifically, we prove an explicit (unreduced) basis of an $\QCCal \QCCbl$-quasi-cyclic product code in terms of the generator matrix in reduced Gröbner basis with respect to the position-over-term order (RGB/POT) form of the $\QCCal$-quasi-cyclic row- and the $\QCCbl$-quasi-cyclic column-code, respectively. This generalizes the work of Burton and Weldon for the generator polynomial of a cyclic product code (where $\QCCal=\QCCbl=1$). Furthermore, we derive the generator matrix in Pre-RGB/POT form of an $\QCCal \QCCbl$-quasi-cyclic product code for two special cases: (i) for $\QCCal=2$ and $\QCCbl=1$, and (ii) if the row-code is a $1$-level $\QCCal$-quasi-cyclic code (for arbitrary $\QCCal$) and $\QCCbl=1$.
For arbitrary $\QCCal$ and $\QCCbl$, the Pre-RGB/POT form of the generator matrix of an $\QCCal \QCCbl$-quasi-cyclic product code is conjectured.

The spectral analysis is applied to the generator matrix of the product of an $\QCCl$-quasi-cyclic and a cyclic code, and we propose a new lower bound on the minimum Hamming distance of a given $\QCCl$-quasi-cyclic code. In addition, we develop an efficient syndrome-based  decoding algorithm for $\QCCl$-phased burst errors with guaranteed decoding radius.
\end{abstract}


\section{Introduction}
The family of quasi-cyclic codes over finite fields is an important class of linear block codes, which is---in contrast to cyclic codes---known to be asymptotically good (see, e.g., Chen~\textit{et al.}~\cite{chen_results_1969}).
Several quasi-cyclic codes have the highest minimum Hamming distance for a given length and dimension (see, e.g., Gulliver and Bhargava~\cite{gulliver_best_1991} as well as Chen's and Grassl's databases~\cite{chen_database_2014,grassl_bounds_2007}). Many good LDPC codes are quasi-cyclic (see, e.g.,~\cite{butler_bounds_2013}) and the connection to convolutional codes was investigated among others in~\cite{solomon_connection_1979,esmaeili_link_1998,lally_algebraic_2006}.

Recent works of Barbier~\textit{et al.}~\cite{barbier_quasi-cyclic_2012, barbier_decoding_2013}, Lally and Fitzpatrick~\cite{lally_algebraic_2001, lally_quasicyclic_2003, lally_algebraic_2006}, Ling and Solé~\cite{ling_algebraic_2001, ling_algebraic_2003, ling_algebraic_2005}, Semenov and Trifonov~\cite{semenov_spectral_2012} and Güneri and Özbudak~\cite{guneri_bound_2012} discuss different aspects of the algebraic structure of quasi-cyclic codes. Although several of these works~\cite{barbier_quasi-cyclic_2012, barbier_decoding_2013,lally_algebraic_2001, lally_quasicyclic_2003, lally_algebraic_2006,ling_algebraic_2001, ling_algebraic_2003, ling_algebraic_2005, semenov_spectral_2012,guneri_bound_2012}  propose new lower bounds on the minimum Hamming distance, their estimates are still far away from the real minimum distance, and therefore, it is an open issue to find better bounds and in addition to develop efficient algebraic decoding approaches.

The work of Wasan~\cite{wasan_quasi_1977} considers quasi-cyclic product codes while investigating the mathematical properties of the wider class of quasi-abelian codes. Some more results were published in a short note by Wasan and Dass~\cite{dass_note_1983}. Koshy proposed a so-called ``circle'' quasi-cyclic product code in~\cite{koshy_quasi-cyclic_1972}.

This work provides the generator matrix of an $\QCCal \QCCbl$-quasi-cyclic product code $\QCCa \otimes \QCCb$ based on the given reduced Gröbner basis (RGB) representation of Lally and Fitzpatrick~\cite{lally_algebraic_2001} of the $\QCCal$-quasi-cyclic row-code $\QCCa$ and the $\QCCbl$-quasi-cyclic column-code $\QCCb$. This generalizes the results of Burton and Weldon~\cite{burton_cyclic_1965} and Lin and Weldon~\cite{lin_further_1970} for the generator polynomial of a cyclic product code (see also~\cite[Chapter 18]{macwilliams_theory_1988}). The generator matrix in Pre-RGB/POT form of an $\QCCal \QCCbl$-quasi-cyclic product code $\QCCa \otimes \QCCb$ is derived for two special cases: (i) for $\QCCal=2$ and $\QCCbl=1$, and (ii) if the row-code $\QCCa$ is a $1$-level $\QCCal$-quasi-cyclic code and $\QCCbl=1$. We conjecture the basis of $\QCCa \otimes \QCCb$ for arbitrary $\QCCal$ and $\QCCbl$.

We apply the spectral analysis of Semenov and Trifonov~\cite{semenov_spectral_2012} to the generator matrix in Pre-RGB/POT form of an $\QCCl$-quasi-cyclic product code $\QCCa \otimes \QCCb$, where $\QCCa$ is an $\QCCl$-quasi-cyclic code and $\QCCb$ is a cyclic code. Moreover, we propose a new lower bound $\NewBCHBound$ on the minimum Hamming distance of a given $\QCCl$-quasi-cyclic code $\QCCa$ via embedding $\QCCa$ into an $\QCCl$-quasi-cyclic product code $\QCCa \otimes \QCCb$. This embedding approach provides an efficient syndrome-based algebraic decoding algorithm that guarantees to decode up to $\lfloor (\NewBCHBound{-}1)/2 \rfloor$ $\QCCl$-phased burst errors. 

The paper is structured as follows. We introduce basic notation, recall relevant parts of the Gröbner basis theory for quasi-cyclic codes of Lally and Fitzpatrick~\cite{lally_algebraic_2001} and the spectral analysis technique of Semenov and Trifonov~\cite{semenov_spectral_2012} in Section~\ref{sec_Preliminaries}. Section~\ref{sec_ProductQCCQCC} covers elementary properties of quasi-cyclic product codes and our main theorem (Thm.~\ref{theo_ProductCodeQCCQCCUnred}) on the (unreduced) basis of an $\QCCal \QCCbl$-quasi-cyclic product code $\QCCa \otimes \QCCb$, in terms of the two given generator matrices in RGB/POT form of the $\QCCal$-quasi-cyclic row-code $\QCCa$ and the $\QCCbl$-quasi-cyclic column-code $\QCCb$, is proven. The generator matrix of a quasi-cyclic product code is derived for two special cases. The first case is a $2$-quasi-cyclic product code of a $2$-quasi-cyclic and a cyclic code and its generator matrix in Pre-RGB/POT form is proposed in Thm.~\ref{theo_ProductCode2-QCCTimesCyclic}. Thm.~\ref{theo_OneLevelQC} gives the RGB form for the second case, i.e., an $\QCCl$-quasi-cyclic product code of a $1$-level $\QCCl$-quasi-cyclic and a cyclic code. The explicit expression of the generator matrix of $\QCCa \otimes \QCCb$ in Pre-RGB/POT form for arbitrary $\QCCal$ and $\QCCbl$ is presumed in Conjecture~\ref{conj_ProductCodeQCCQCC}, which we verified through reducing the unreduced basis of several examples.

Although we could prove the RGB/POT form of the generator matrix of a quasi-cyclic product code only for the previously mentioned cases (Thm.~\ref{theo_ProductCode2-QCCTimesCyclic} and Thm.~\ref{theo_OneLevelQC}), we perform the spectral analysis for the instance of an $\QCCl$-quasi-cyclic product code $\QCCa \otimes \QCCb$, where $\QCCa$ is an $\QCCl$-quasi-cyclic and $\QCCb$ is a cyclic code in Section~\ref{sec_SpectralAnalysis} based on Conjecture~\ref{conj_ProductCodeQCCQCC}. The new lower bound $\NewBCHBound$ is proposed in Section~\ref{subsec_BoundingDistance}. Section~\ref{sec_Decoding} contains our syndrome-based decoding algorithm with guaranteed $\QCCl$-phased burst error-correcting radius $\lfloor (\NewBCHBound-1)/2 \rfloor$. We conclude the paper in Section~\ref{sec_conclusion}.

\section{Preliminaries} \label{sec_Preliminaries}
\subsection{Notation and Reduced Gröbner Basis (RGB)}
Let $\F{q}$ denote the finite field of order $q$, $\Fx{q}$ the polynomial ring over $\F{q}$ with indeterminate $X$, and $\F{q}^n$ the linear vector space over $\F{q}$  of dimension $n$. The entries of a vector $\vec{v} \in \F{q}^n$ are indexed from zero to $n{-}1$, i.e., $\vec{v} = (v_0 \ v_1 \ \cdots \ v_{n-1})$.
For two vectors $\vec{v}, \vec{w} \in \F{q}^n$, the scalar product $\sum_{i=0}^{n-1} v_i w_i$ is denoted by $\vec{v} \circ \vec{w}$. For two positive integers $a, b$ with $b > a$ the set of integers $\{a,a+1,\dots,b-1\}$ is denoted by $\interval{a,b}$ and we define the short-hand notation $\interval{b} \defeq \interval{0,b}$. An $m \times n$ matrix $\mathbf{M} \in \F{q}^{m \times n}$ is denoted as $\mathbf{M} = (m_{i,j})_{i \in \interval{m}}^{j \in \interval{n}}$ or where the size follows from the context, we use the short-hand notation $(m_{i,j})$.

A linear $\LIN{\QCCl \cdot \QCCm}{\QCCk}{\QCCd}{q}$ code $\QCC$ of length $\QCCl \QCCm$, dimension $\QCCk$, and minimum Hamming distance $\QCCd$ over $\F{q}$ is $\QCCl$-quasi-cyclic if every cyclic shift by $\QCCl$ of a codeword is again a codeword of $\QCC$, more explicitly if:
\begin{align*}
(c_{0,0} \ \cdots \ c_{\QCCl-1,0} \ c_{0,1} \  \cdots \ c_{\QCCl-1,1} \ \cdots \ c_{0,\QCCm-1} & \ \cdots \ c_{\QCCl-1,\QCCm-1}) \in \QCC \Rightarrow \\
& (c_{0,\QCCm-1} \ \cdots \ c_{\QCCl-1,\QCCm-1} \ c_{0,0} \ \cdots \ c_{\QCCl-1,0} \ \cdots \ c_{0,\QCCm-2} \ \cdots \ c_{\QCCl-1,\QCCm-2}) \in \QCC.
\end{align*}
We can represent a codeword of an $\LIN{\QCCl \cdot \QCCm}{\QCCk}{\QCCd}{q}$ $\QCCl$-quasi-cyclic code $\QCC$ as $\mathbf{c}(X) = (c_0(X) \ c_1(X) \ \cdots \ c_{\QCCl-1}(X)) \in \Fx{q}^{\ell} $, where each entry is given by
\begin{equation} \label{eq_QCCComponents}
c_j(X) \defeq \sum_{i=0}^{\QCCm-1} c_{j,i} X^{i}, \quad \forall j \in \interval{\QCCl}.
\end{equation}
Then, the defining property of the $\QCCl$-quasi-cyclic code $\QCC$ is that it is closed under multiplication by $X$ modulo $(X^{\QCCm}-1)$ in each entry.
\begin{lemma}[Codeword Representation: Vector to Univariate Polynomial] \label{lem_VectorToUnivariatePoly}
Let $(c_0(X) \ c_1(X) \ \cdots \ c_{\QCCl-1}(X))$ be a codeword of an $\LIN{\QCCl \cdot \QCCm}{\QCCk}{\QCCd}{q}$ $\QCCl$-quasi-cyclic code $\QCC$, where the entries are defined as in \eqref{eq_QCCComponents}. Then a codeword in $\QCC$, represented as one univariate polynomial of degree smaller than $\QCCl \QCCm$, is
\begin{equation} \label{eq_VectorToUnivariatePoly}
c(X) = \sum_{j=0}^{\QCCl-1} c_j(X^{\QCCl})X^j.
\end{equation}
\end{lemma}
\begin{proof}
Substituting~\eqref{eq_QCCComponents} into~\eqref{eq_VectorToUnivariatePoly} leads to:
\begin{align*}
c(X) & = \sum_{j=0}^{\QCCl-1} c_j(X^{\QCCl})X^j = \sum_{j=0}^{\QCCl-1} \sum_{i=0}^{\QCCm-1} c_{j,i} X^{i\QCCl+j}.
\end{align*}
\end{proof}
Lally and Fitzpatrick~\cite{lally_construction_1999, lally_algebraic_2001} showed that an $\QCCl$-quasi-cyclic code $\QCC$ can be viewed as an $R$-submodule of the algebra $R^{\QCCl}$, where $R = \Fx{q}/\langle X^{\QCCm}-1 \rangle$. The code $\QCC$ is the image of an $\Fx{q}$-submodule $\tilde{\QCC}$ of $\Fx{q}^{\QCCl}$ containing
\begin{equation*}
\basis = \left \langle (X^{\QCCm}-1)\mathbf{e}_j, j \in \interval{\QCCl} \right \rangle,
\end{equation*}
where $\mathbf{e}_j \in \Fx{q}^{\QCCl}$ is the standard basis vector with one in position $j$ and zero elsewhere under the natural homomorphism
\begin{equation} \label{eq_Homorphism}
\begin{split}
\phi: \; \Fx{q}^{\QCCl} & \rightarrow  R^{\QCCl} \\
 (c_0(X) \ \cdots \ c_{\QCCl-1}(X)) & \mapsto (c_0(X) + \langle X^{\QCCm}-1 \rangle \  \cdots \ c_{\QCCl-1}(X) +\langle X^{\QCCm}-1 \rangle ).
\end{split}
\end{equation}
The submodule has a generating set of the form $\{ \mathbf{u}_i, i \in \interval{z}, (X^{\QCCm}-1)\mathbf{e}_j, j \in \interval{\QCCl} \}$, where $\mathbf{u}_i \in \Fx{q}^{\QCCl}$ and $z \leq \QCCl$ (see, e.g.,~\cite[Chapter 5]{cox_using_1998} for further information) and can be represented as a matrix with entries in $\Fx{q}$:
\begin{equation} \label{eq_GeneratorWithBasis}
\mathbf{U}(X) = 
\begin{pmatrix}
u_{0,0}(X) & u_{0,1}(X) & \cdots & u_{0,\QCCl-1}(X) \\
u_{1,0}(X) & u_{1,1}(X) & \cdots & u_{1,\QCCl-1}(X) \\
 \vdots & \vdots & \ddots & \vdots \\
u_{z-1,0}(X) & u_{z-1,1}(X) & \cdots & u_{z-1,\QCCl-1}(X) \\
X^{\QCCm}-1 &  &  \\
& X^{\QCCm}-1 & \multicolumn{2}{c}{\bigzero} \\
\multicolumn{2}{c}{\bigzero} & \ddots \\
& & & X^{\QCCm}-1
\end{pmatrix}.
\end{equation}
Every matrix $\mathbf{U}(X)$ as in~\eqref{eq_GeneratorWithBasis} of an $\QCCl$-quasi-cyclic code $\QCC$ can be transformed to a reduced Gröbner basis (RGB) with respect to the position-over-term order (POT) in $\Fx{q}^{\QCCl}$ (as shown in~\cite{lally_construction_1999, lally_algebraic_2001}).
A basis in RGB/POT form can be represented by an upper-triangular $\ell \times \ell$ matrix with entries in $\Fx{q}$ as follows: 
\begin{equation} \label{def_GroebBasisMatrix}
\mathbf{G}(X) =
\begin{pmatrix}
g_{0,0}(X) & g_{0,1}(X) & \cdots  & g_{0,\QCCl-1}(X) \\
 & g_{1,1}(X) & \cdots & g_{1,\QCCl-1}(X) \\
\multicolumn{2}{c}{\bigzero}& \ddots & \vdots \\
 &  & & g_{\QCCl-1,\QCCl-1}(X)
\end{pmatrix},
\end{equation}
where the following conditions must be fulfilled:\\

\begin{tabular}[htb]{lrcll}
C1: & $g_{i,j}(X)$ & $=$ & $0,$ & $\forall 0 \leq j < i < \QCCl$, \\
C2: & $\deg g_{j,i}(X)$ & $<$ & $ \deg g_{i,i}(X),$ & $ \forall j < i, i \in \interval{\QCCl}$,\\
C3: & $g_{i,i}(X)$ & $|$ & $ (X^{\QCCm}-1),$ & $\forall i \in \interval{\QCCl}$,\\
C4: & if $g_{i,i}(X)$ & $=$ & $X^{\QCCm}-1$ then \\ 
& $g_{i,j}(X)$ & $=$ & $0,$ & $ \forall j \in \interval{i+1, \QCCl}$.
\end{tabular}\\

We refer to these conditions as RGB/POT conditions C1--C4 throughout this paper and refer to the unreduced representation as in~\eqref{eq_GeneratorWithBasis} if necessary.
The rows of $\mathbf{G}(X)$  with $g_{i,i}(X) \neq X^{\QCCm}-1$ (i.e., the rows that do not map to zero under $\phi$ as in~\eqref{eq_Homorphism}) are called the reduced generating set of the quasi-cyclic code $\QCC$.
Let $k_j= \QCCm - \deg g_{j,j}(X)$ for all $j \in \interval{\QCCl}$. A codeword of $\QCC$ can be represented as $\mathbf{c}(X) = \mathbf{i}(X) \mathbf{G}(X)$, where $\vec{i}(X) = (i_0(X) \ i_1(X) \ \cdots \ i_{\QCCl-1}(X))$ and $\deg i_j(X) < k_j, \, \forall j \in \interval{\QCCl}$. The dimension of $\QCC$ is $\QCCk = \QCCm \QCCl - \sum_{j=0}^{\QCCl-1} \deg g_{j,j}(X)$.
For $\QCCl=1$, the generator matrix $\mathbf{G}(X)$ in RGB/POT form as in~\eqref{def_GroebBasisMatrix} becomes the well-known generator polynomial of a cyclic code of degree $\QCCm-\QCCk$. In this paper we consider the single-root case, i.e., $\gcd(\QCCm, \text{char}(\F{q}))=1$.

We recall the following definition (see \cite[Thm. 3.2]{lally_construction_1999}).
\begin{definition}[$\level$-level Quasi-Cyclic Code {\cite[Thm. 3.2]{lally_construction_1999}}] \label{def_LevelQC}
We call an $\LIN{\QCCl \cdot \QCCm}{\QCCk}{\QCCd}{q}$ $\QCCl$-quasi-cyclic code $\QCC$ an $\level$-level quasi-cyclic code if there is an index $\level \in \interval{\QCCl}$ for which the RGB/POT matrix as defined in~\eqref{def_GroebBasisMatrix} is such that $g_{\level-1,\level-1}(X) \neq X^{\QCCm}-1$ and $g_{\level,\level}(X) = \dots = g_{\QCCl-1,\QCCl-1}(X) = X^{\QCCm}-1$.
\end{definition}
Furthermore, the generator matrix in RGB/POT form of a $1$-level $\QCCl$-quasi-cyclic code as in Def.~\ref{def_LevelQC} is stated in the following corollary.
\begin{corollary}[$1$-level Quasi-Cyclic Code {\cite[Corollary 3.3]{lally_construction_1999}}] \label{cor_OneLevelQC}
The generator matrix in RGB/POT form of an $\LIN{\QCCl \cdot \QCCm}{\QCCk}{\QCCd}{q}$ $1$-level $\QCCl$-quasi-cyclic code $\QCC$ has the following form:
\begin{equation*}
\mathbf{G}(X) =
\begin{pmatrix}
 g(X) & g(X) f_{1}(X)  & \cdots  & g(X)f_{\QCCl-1}(X)
\end{pmatrix},
\end{equation*}
where $g(X) \mid (X^{\QCCm}-1)$, $ \deg g(X) = \QCCm-\QCCk$, and $f_{1}(X), \dots, f_{\QCCl-1}(X) \in \Fx{q}$.
\end{corollary}

\subsection{Spectral Analysis of Quasi-Cyclic Codes} \label{Subsec_SpectralAnalysis}
Let $\mathbf{G}(X)$ be the upper-triangular generator matrix of a given $\LIN{\QCCl \cdot \QCCm}{\QCCk}{\QCCd}{q}$ $\QCCl$-quasi-cyclic code $\QCC$ in RGB/POT form as in~\eqref{def_GroebBasisMatrix}. Let $\alpha \in \F{q^{\QCCs}}$ be an \ith{\QCCm} root of unity.
An eigenvalue $\eigenvalue{i}$ of $\QCC$ is defined to be a root of $\det(\mathbf{G}(X))$, i.e., a root of $\prod_{j=0}^{\QCCl-1} g_{j,j}(X)$.
The \textit{algebraic} multiplicity of $\eigenvalue{i}$ is the largest integer $ \mult{i} $ such that 
$(X-\eigenvalue{i})^{\mult{i}} \mid \det(\mathbf{G}(X)).$
Semenov and Trifonov~\cite{semenov_spectral_2012} defined the \textit{geometric} multiplicity of an eigenvalue $\eigenvalue{i}$ as the dimension of the right kernel of the matrix $\mathbf{G}(\eigenvalue{i})$, i.e., the dimension of the solution space of the homogeneous linear system of equations:
\begin{equation} \label{eq_eigenvectors}
\mathbf{G}(\eigenvalue{i}) \eigenvector =  \textbf{0}.
\end{equation}
The solution space of~\eqref{eq_eigenvectors} is called the right kernel eigenspace and it is denoted by $\eigenspace[i]$. Furthermore, it was shown that, for a matrix $\mathbf{G}(X) \in \Fx{q}^{\QCCl \times \QCCl}$ in RGB/POT form, the algebraic multiplicity $\mult{i}$ of an eigenvalue $\eigenvalue{i}$ equals the geometric multiplicity~\cite[Lemma 1]{semenov_spectral_2012}.
\begin{definition}[Pre-RGB/POT Form] \label{def_PreRGBPOTForm}
A generator matrix $\bar{\mathbf{G}}(X)$ of $\QCC$ that satisfies RGB/POT Conditions C1, C3 and C4, but not C2, is called a matrix in Pre-RGB/POT form. More explicitly, the generator matrix has the following form:
\begin{equation} \label{eq_PRE-RGBPOTForm}
\bar{\mathbf{G}}(X) =
\begin{pmatrix}
g_{0,0}(X) & \bar{g}_{0,1}(X) & \cdots  & \bar{g}_{0,\QCCl-1}(X) \\
 & g_{1,1}(X) & \cdots & \bar{g}_{1,\QCCl-1}(X) \\
\multicolumn{2}{c}{\bigzero}& \ddots & \vdots \\
 &  & & g_{\QCCl-1,\QCCl-1}(X)
\end{pmatrix}, 
\end{equation}
where the entries of $\bar{\mathbf{G}}(X)$ that can be different from their counterparts in the RGB/POT form, are marked by a bar.
\end{definition}

\begin{lemma}[Equivalence of the Spectral Analysis of a Matrix in Pre-RGB/POT Form] \label{lem_EquivalenceSpectralAnalysis}
Let $\mathbf{G}(X)$ be an $\QCCl \times \QCCl$ generator matrix of an $\QCCl$-quasi-cyclic code $\QCC$ in RGB/POT form as in~\eqref{def_GroebBasisMatrix} and let $\bar{\mathbf{G}}(X)$ be a generator matrix of the same code in Pre-RGB/POT form as in Definition~\ref{def_PreRGBPOTForm}.

Let $\lambda_i$ be an eigenvalue of $\mathbf{G}(X)$. Then, the right kernels of $\mathbf{G}(\lambda_i)$ and $\bar{\mathbf{G}}(\lambda_i)$ are equal, i.e., the (algebraic and geometric) multiplicity and the corresponding eigenvalues are identical.
\end{lemma}
\begin{proof}
To reduce the matrix $\bar{\mathbf{G}}(X)$ to $\mathbf{G}(X)$ only linear transformations in $\Fx{q}$, i.e., linear combinations of rows are necessary and therefore the right kernels of $\mathbf{G}(\lambda_i)$ and $\bar{\mathbf{G}}(\lambda_i)$ are the same.
\end{proof}
Moreover, Semenov and Trifonov~\cite{semenov_spectral_2012} gave an explicit construction of the parity-check matrix of an $\LIN{\QCCl \cdot \QCCm}{\QCCk}{\QCCd}{q}$ $\QCCl$-quasi-cyclic code $\QCC$ and proved a BCH-like~\cite{bose_class_1960, hocquenghem_codes_1959} lower bound on the minimum Hamming distance $\QCCd$ (see Thm.~\ref{theo_SemenovTrifonovBound}) using the parity-check matrix and the so-called eigencode. We generalize their approach in Section~\ref{sec_SpectralAnalysis}, but do not explicitly need the parity-check matrix for the proof, though the eigencode is still needed.
\begin{definition}[Eigencode] \label{def_eigencode}
Let $\eigenspace{} \subseteq \F{q^{\QCCs}}^{\QCCl}$ be an eigenspace. Define the $\LIN{\ECn=\QCCl}{\ECk}{\ECd}{q}$ eigencode corresponding to $\eigenspace{}$ by
\begin{equation*} 
\eigencode{\eigenspace{}} \defeq \left \lbrace \vec{c} \in \F{q}^{\QCCl} \mid \forall \eigenvector \in \eigenspace : \vec{v} \circ \vec{c} = 0  \right\rbrace.
\end{equation*}
\end{definition}
If there exists an eigenvector $\eigenvector = (\eigenvector[0] \ \eigenvector[1] \ \cdots \ \eigenvector[\QCCl-1]) \in \eigenspace$ with entries $\eigenvector[0], \eigenvector[1], \dots, \eigenvector[\QCCl-1]$ that are linearly independent over $\F{q}$, then $\eigencode{\eigenspace{}} = \{ (0 \ 0 \ \cdots \ 0) \} $ and $\ECd$ is infinity.

To describe quasi-cyclic codes explicitly, we need to recall the following facts related to cyclic codes. A $q$-cyclotomic coset $\coset{i}$ is defined as:
\begin{equation} \label{eq_cyclotomiccoset}
 \coset{i} \defeq \Big\{ iq^j \mod \QCCm \, \vert \, j \in \interval{a} \Big\},
\end{equation}
where $a$ is the smallest positive integer such that $iq^{a} \equiv i \bmod \QCCm$. 
The minimal polynomial in $\Fx{q}$ of the element $\alpha^i \in \F{q^{\QCCs}}$ is given by
\begin{equation} \label{eq_MinPoly}
\minpoly{\alpha^i} = \prod_{j \in \coset{i} } (X-\alpha^j).
\end{equation}

\section{Quasi-Cyclic Product Codes} \label{sec_ProductQCCQCC}
In this section we consider an $\QCCal \QCCbl$-quasi-cyclic product code $\QCCa \otimes \QCCb$, where the symbol $\otimes$ stems from the fact that a generator matrix with entries in $\F{q}$ of $\QCCa \otimes \QCCb$ is the Kronecker product of the generator matrices (with entries in $\F{q}$) of $\QCCa$ and $\QCCb$ (see, e.g., \cite[Ch. 18. §2]{macwilliams_theory_1988}). In the following, let $\QCCa$ be an $\LIN{\QCCan = \QCCal \cdot \QCCam}{\QCCak}{\QCCad}{q}$ $\QCCal$-quasi-cyclic code generated by the following matrix in RGB/POT form as defined in~\eqref{def_GroebBasisMatrix}:
\begin{equation} \label{eq_GroebMatrixCodeA}
\genmat[A] = 
\begin{pmatrix} 
\gen[A][0][0] & \gen[A][0][1] & \cdots & \gen[A][0][\QCCal-1]\\
 & \gen[A][1][1] & \cdots & \gen[A][1][\QCCal-1] \\
\multicolumn{2}{c}{\bigzero}& \ddots & \vdots \\
 &  & & \gen[A][\QCCal-1][\QCCal-1] 
\end{pmatrix},
\end{equation}
and let $\QCCb$ be an $\LIN{\QCCbn = \QCCbl \cdot \QCCbm}{\QCCbk}{\QCCbd}{q}$ $\QCCbl$-quasi-cyclic code with generator matrix in RGB/POT form:
\begin{equation} \label{eq_GroebMatrixCodeB}
\genmat[B] = 
\begin{pmatrix}
\gen[B][0][0] & \gen[B][0][1] & \cdots & \gen[B][0][\QCCbl-1]\\
 & \gen[B][1][1] & \cdots & \gen[B][1][\QCCbl-1] \\
\multicolumn{2}{c}{\bigzero}& \ddots & \vdots \\
 &  & & \gen[B][\QCCbl-1][\QCCbl-1] 
\end{pmatrix}.
\end{equation}
We assume throughout the paper that $\gcd(\QCCan, \CYCbn) = 1$. Let two integers $\inta$ and $\intb$ be such that 
\begin{equation} \label{eq_BEzoutRel}
\inta \QCCan + \intb \CYCbn = 1.
\end{equation}
A codeword $c(X) \in \Fx{q}$ of the $\LIN{\QCCn = \QCCal \QCCbl \cdot \QCCam \QCCbm}{\QCCk = \QCCak \QCCbk}{\QCCd = \QCCad \QCCbd}{q}$ product code $\QCCa \otimes \QCCb$ can then be obtained from the $\QCCbn \times \QCCan$ matrix $(m_{i,j})_{i \in \interval{\QCCbn}}^{j \in \interval{\QCCan}}$ representation, where each row is in $\QCCa$ and each column is in $\QCCb$, as follows:
\begin{equation} \label{eq_OneUnivariatePolyProduct}
c(X) \equiv \sum_{i=0}^{\QCCbn-1} \sum_{j=0}^{\QCCan-1}  m_{i,j} X^{\map{i}{j}}  \mod (X^{\QCCn}-1), 
\end{equation}
where we give $\map{i}{j}$ in Lemma~\ref{lem_MappingToUnivariatePolyQCC}. This mapping was stated by Wasan in~\cite{wasan_quasi_1977} and generalizes the result of Burton and Weldon~\cite[Thm. I]{burton_cyclic_1965} for a cyclic product code to the case of an $\QCCal \QCCbl$-quasi-cyclic product code $\QCCa \otimes \QCCb$.
\begin{lemma}[Mapping to a Univariate Polynomial~\cite{wasan_quasi_1977}] \label{lem_MappingToUnivariatePolyQCC}
Let $\QCCa$ be an $\QCCal$-quasi-cyclic code of length $\QCCan$ and let $\QCCb$ be an $\QCCbl$-quasi-cyclic code of length $\QCCbn$. The product code $\QCCa \otimes \QCCb$ is an $\QCCal \QCCbl$-quasi-cyclic code of length $\QCCn = \QCCan \QCCbn$ if $\gcd(\QCCan, \QCCbn) = 1$.
\end{lemma}
\begin{proof}
A codeword of the $\LIN{\QCCan \QCCbn}{\QCCak \QCCbk}{\QCCad \QCCbd}{q}$ product code $\QCCa \otimes \QCCb$ can be represented by an $\QCCbn \times \QCCan$ matrix $(m_{i,j})_{i \in \interval{\QCCbn}}^{j \in \interval{\QCCan}}$, where each row is a codeword of $\QCCa$ and each column is a codeword of $\QCCb$. The entries of the matrix $(m_{i,j})$ in the \ith{i} row and \ith{j} column are mapped to the coefficients of the codeword by:
\begin{equation} \label{def_MappingMatrixPolyQCCQCC}
\map{i}{j} \defeq i \inta \QCCan \QCCal + j \intb \QCCbn \QCCbl  \mod \QCCn,
\end{equation}
where $i \in \interval{\QCCbn}$ and $j \in \interval{\QCCan}$.
In order to prove that the product code $\QCCa \otimes \QCCb$ is $\QCCal \QCCbl$-quasi-cyclic it is sufficient to show that a shift by $\QCCal \QCCbl$ of a codeword in $\QCCa \otimes \QCCb$ serialized to a univariate polynomial by~\eqref{def_MappingMatrixPolyQCCQCC} is again a codeword in $\QCCa \otimes \QCCb$. This will be true if a shift by $\QCCal$ in every row and a shift by $\QCCbl$ in every column correspond to an $\QCCal \QCCbl$-quasi-cyclic shift of the univariate codeword obtained by~\eqref{def_MappingMatrixPolyQCCQCC}, which is indeed the case:
\begin{align*}
\map{i+\QCCbl}{j+\QCCal} & \equiv (i+\QCCbl) \inta \QCCan \QCCal + (j+\QCCal) \intb \QCCbn \QCCbl  \mod \QCCn \\
& \equiv i \inta \QCCan \QCCal + j \intb \QCCbn \QCCbl + \QCCal \QCCbl (\inta \QCCan  + \intb \QCCbn) \mod \QCCn \\
& \equiv \map{i}{j} + \QCCal \QCCbl \mod \QCCn.
\end{align*}
\end{proof}
Instead of representing a codeword in $\QCCa \otimes \QCCb$ as one univariate polynomial in $\Fx{q}$ as in~\eqref{eq_OneUnivariatePolyProduct}, we want to represent it as a vector of $\QCCal \QCCbl$ univariate polynomials in $\Fx{q}$ (as in Lemma~\ref{lem_VectorToUnivariatePoly}) to obtain an explicit expression of the basis of the $\QCCal \QCCbl$-quasi-cyclic product code $\QCCa \otimes \QCCb$. 
\begin{lemma}[Mapping to $\QCCal \QCCbl$ Univariate Polynomials] \label{lem_MappingUnivariateQCCQCC}
Let $\QCCa$ be an $\QCCal$-quasi-cyclic code of length $\QCCan = \QCCal \QCCam $ and let $\QCCb$ be an $\QCCbl$-quasi-cyclic code of length $\QCCbn = \QCCbl \QCCbm$. Let $\QCCl = \QCCal \QCCbl$, $\QCCm = \QCCam \QCCbm$, and $\QCCn = \QCCan \QCCbn$. Let $(m_{i,j})_{i \in \interval{\QCCbn}}^{j \in \interval{\QCCan}}$ be a codeword of the $\QCCl$-quasi-cyclic product code $\QCCa \otimes \QCCb$, where each row is in $\QCCa$ and each column is in $\QCCb$. 

Define $\QCCl$ univariate polynomials as:
\begin{equation} \label{eq_UnivariateFinalQCCQCC}
c_{g,h}(X) \equiv X^{\submap{g}{h}} \cdot \sum_{i=0}^{\QCCbm-1} \sum_{j=0}^{\QCCam-1} m_{i\QCCbl+g,j\QCCal+h} X^{ \mapb{i}{j} } \mod (X^{\QCCm} - 1), \quad \forall g \in \interval{\QCCbl}, h \in \interval{\QCCal},
\end{equation}
with
\begin{align} 
\submap{g}{h} & = g(-\intb\QCCbm ) + h(-\inta \QCCam) \mod \QCCm, \label{eq_ShiftCodeword} \\
\mapb{i}{j} &= i \inta \QCCan + j \intb \QCCbn \mod \QCCm. \label{eq_MappingSubCodewordQCCQCC}
\end{align}
Then the codeword $c(X) \in \QCCa \otimes \QCCb$ corresponding to $(m_{i,j})_{i \in \interval{\QCCbn}}^{j \in \interval{\QCCan}}$ is given by:
\begin{equation} \label{eq_ProofMatrixToUnivariateFirst} 
c(X)  \equiv  \sum_{g=0}^{\QCCbl-1} \sum_{h=0}^{\QCCal-1} c_{g,h} (X^{\QCCal \QCCbl}) X^{g\QCCal + h \QCCbl} \mod (X^{\QCCn}-1).
\end{equation}
\end{lemma}
\begin{proof}
We have:
\begin{align}
y & \equiv  \mapb{i}{j} + \submap{g}{h}  \mod \QCCm. \nonumber \\
&  \qquad \qquad \Leftrightarrow \nonumber \\
\QCCal \QCCbl y & \equiv \QCCal \QCCbl ( \mapb{i}{j}+  \submap{g}{h} )  \mod \QCCn . \label{eq_MovingFromBigToSmallMatrixQCCQCC}
\end{align}
With $\inta \QCCan-1 = -\intb \QCCbn = -\intb \QCCbm \QCCbl$ and $\intb \QCCbn-1 = -\inta \QCCan = -\inta \QCCam \QCCal$, and using~\eqref{eq_ShiftCodeword}, we can rewrite~\eqref{eq_MovingFromBigToSmallMatrixQCCQCC} to:
\begin{align}
\QCCal \QCCbl ( \mapb{i}{j} + \submap{g}{h} ) & \equiv \QCCal \QCCbl \mapb{i}{j} + g \QCCal (- \intb \QCCbm \QCCbl ) + h \QCCbl ( -\inta \QCCam \QCCal)  \mod \QCCn \nonumber \\
& \equiv \QCCal \QCCbl \mapb{i}{j} + g \QCCal (\inta \QCCan  - 1 ) + h \QCCbl (\intb \QCCbn  - 1) \mod \QCCn.  \label{eq_SpreadedCodewordsQCCQCC}
\end{align}
With $\mapb{i}{j}$ as in~\eqref{eq_MappingSubCodewordQCCQCC} and $\map{i}{j}$ as in~\eqref{def_MappingMatrixPolyQCCQCC}, we get from~\eqref{eq_SpreadedCodewordsQCCQCC}:
\begin{align}
\QCCal \QCCbl (i \inta \QCCan  + j \intb \QCCbn) + g \QCCal (\inta \QCCan  - 1 ) + h \QCCbl (\intb \QCCbn  - 1) & \equiv (i \QCCbl+g) \inta \QCCan \QCCal + (j \QCCal+h) \intb \QCCbn \QCCbl - g \QCCal- h \QCCbl  \mod \QCCn  \nonumber \\
& \equiv \map{i \QCCbl+g}{j\QCCal+h} - g \QCCal- h \QCCbl  \mod \QCCn. \label{eq_MappingToOneCodewordBack}
\end{align}
Inserting~\eqref{eq_UnivariateFinalQCCQCC} into~\eqref{eq_ProofMatrixToUnivariateFirst} and using the result \eqref{eq_MappingToOneCodewordBack} for the manipulations of the exponents leads to:
\begin{align}
c(X) & \equiv \sum_{g=0}^{\QCCbl-1} \sum_{h=0}^{\QCCal-1} \sum_{i=0}^{\QCCbm-1} \sum_{j=0}^{\QCCam-1}  m_{i\QCCbl+g,j\QCCal+h} X^{\map{i\QCCbl+g}{j\QCCal+h}}  \mod (X^{\QCCn}-1). \label{eq_ProofMatrixToUnivariate}
\end{align}
With $i^{\prime}=i\QCCbl+g$ and $j^{\prime} = j\QCCal+h$, we obtain from~\eqref{eq_ProofMatrixToUnivariate}:
\begin{align*}
c(X) & \equiv \sum_{i^{\prime}=0}^{\QCCbn-1} \sum_{j^{\prime}=0}^{\QCCan-1}  m_{i^{\prime},j^{\prime}} X^{\map{i^{\prime}}{j^{\prime}}}  \mod (X^{\QCCn}-1),
\end{align*} 
which coincides with the expression as in~\eqref{eq_OneUnivariatePolyProduct}.
\end{proof}
The mapping $\mapb{i}{j}$ as in~\eqref{eq_MappingSubCodewordQCCQCC} of the $\QCCl$ submatrices $(m_{i\QCCbl,j\QCCal}), (m_{i\QCCbl,j\QCCal+1}), \dots, (m_{i\QCCbl+\QCCbl-1,j\QCCal+\QCCal-1}) \in \F{q}^{\QCCbm \times \QCCam}$ to the $\QCCl$ univariate polynomials $c_{0,0}(X), c_{0,1}(X), \dots, c_{\QCCbl-1, \QCCal-1}(X) $ is the same as the one used to map the codeword of a cyclic product code of length $\QCCam \QCCbm$ from its matrix representation to the polynomial representation (see \cite[Thm. 1]{burton_cyclic_1965} and Fig.~\ref{fig_QCCvert}).
We illustrate the mapping of the matrix to the polynomial representation of a codeword of an $\QCCl$-quasi-cyclic product code as discussed in Lemma~\ref{lem_MappingToUnivariatePolyQCC} and Lemma~\ref{lem_MappingUnivariateQCCQCC} in the following example.
\begin{example}[$6$-Quasi-Cyclic Product Code] \label{ex_6QCCQCC}
Let $\QCCa$ be a $2$-quasi-cyclic code of length $\QCCan = 2 \cdot 5 = 10$ and let $\QCCb$ be a $3$-quasi-cyclic of length $\QCCbn= 3 \cdot 3 = 9$. Let $\inta = 1$ and $\intb = -1$, such that~\eqref{eq_BEzoutRel} holds. For the purpose of this illustration, the field size $q$ is irrelevant, but we assume that $\gcd(\QCCan,q) = \gcd(\QCCbn,q) = 1$.
Fig.~\ref{fig_QCCQCC} contains three different representations of a codeword of the $6$-quasi-cyclic product code $\QCCa \otimes \QCCb$ of length $90$. The $9 \times 10$ matrix $(m_{i,j})_{i \in \interval{9}}^{j \in \interval{10}}$, where each row is a codeword in $\QCCa$ and each column is a codeword in $\QCCb$, is illustrated in Fig.~\ref{fig_QCCQCCfull}. The entries of $(m_{i,j})$ contain the indices of the coefficients if the matrix $(m_{i,j})$ is mapped to a univariate polynomial as given in~\eqref{eq_OneUnivariatePolyProduct}. The color of a code symbol indicates the membership of an entry when the codeword in $\QCCa \otimes \QCCb$ is represented as six univariate polynomials as stated in Lemma~\ref{lem_MappingUnivariateQCCQCC}. The corresponding six $3 \times 5$ submatrices $(m_{i3,j2})$, $(m_{i3+1,j2})$, $(m_{i3+2,j2})$, $(m_{i3,j2+1})$, $(m_{i3+1,j2+1})$, $(m_{i3+2,j2+1}) \in \F{q}^{3 \times 5}$ are depicted separately twice in Fig.~\ref{fig_QCCvertintermed} and in Fig.~\ref{fig_QCCvert}, respectively. Both figures contain different indices of the six univariate polynomials as outlined in the corresponding captions.
\newcommand{\mysize}{6mm}
\newcommand{\mysizeb}{6.5mm}
\begin{figure}[htb]
\centering
\subfigure[Illustration of $\map{i}{j}$ as in~\eqref{def_MappingMatrixPolyQCCQCC} for $\inta=1$, $\QCCal = 2$, $\QCCam = 5$ and $\intb=-1$, $\QCCbl=3$, $\QCCbm=3$. The entry of the $(3\cdot 3) \times (2 \cdot 5)$ matrix $(m_{i,j}) \in \QCCa \otimes \QCCb $ in the $\ith{i}$ row and the $\ith{j}$ column is the $\ith{\map{i}{j}}$ coefficient of the univariate polynomial of degree less than $90$ representing a codeword of $\QCCa \otimes \QCCb$.]{\resizebox{.7\textwidth}{!}{\includegraphics{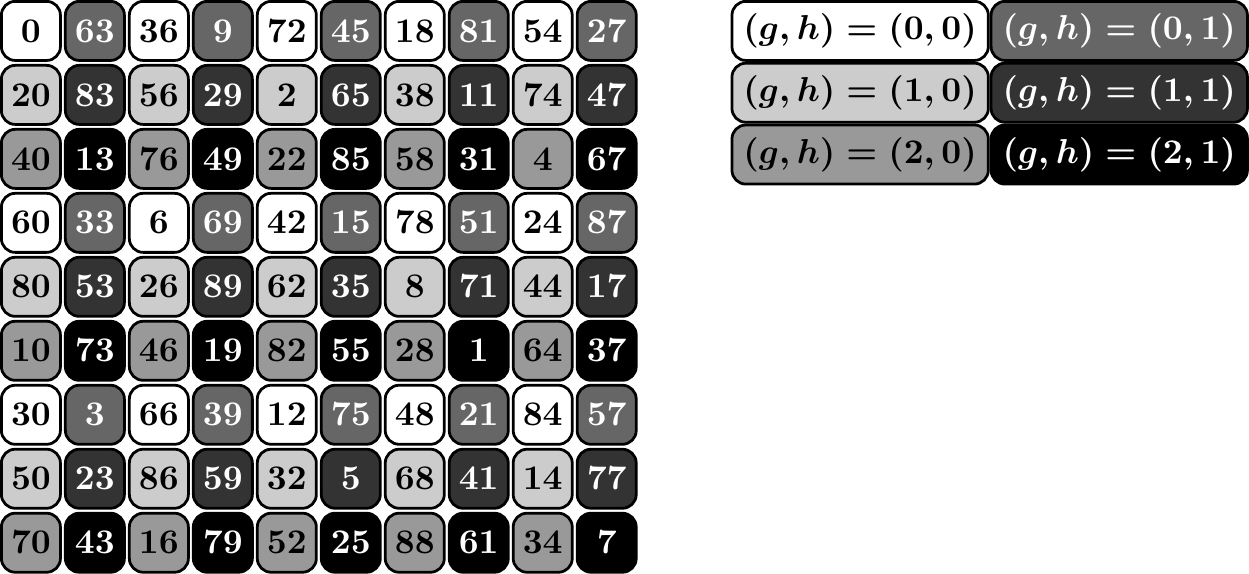}}\label{fig_QCCQCCfull}}\\
 \subfigure[The three submatrices $(m_{i3,j2})$, $(m_{i3+1,j2})$, $(m_{i3+2,j2})$ on the left and the three submatrices $(m_{i3,j2+1})$, $(m_{i3+1,j2+1})$, $(m_{i3+2,j2+1})$ on the right (for $i \in \interval{3}$ and $j \in \interval{5}$) of the $9 \times 10$ matrix $(m_{i,j})$ as given in Fig.~\ref{fig_QCCQCCfull}. The entry of the $\ith{i}$ row and the $\ith{j}$ column of $(m_{i3+g,j2+h})$ is the value $\map{i3+g}{j2+h}$ for all $g \in \interval{3}$ and $h \in \interval{2}$.]{\tikzsetnextfilename{CyclicProductCodesVert1intermediate}\resizebox{.18\textwidth}{!}{\includegraphics{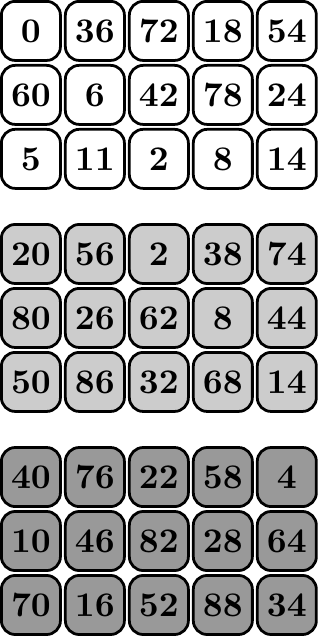}}\resizebox{.18\textwidth}{!}{\includegraphics{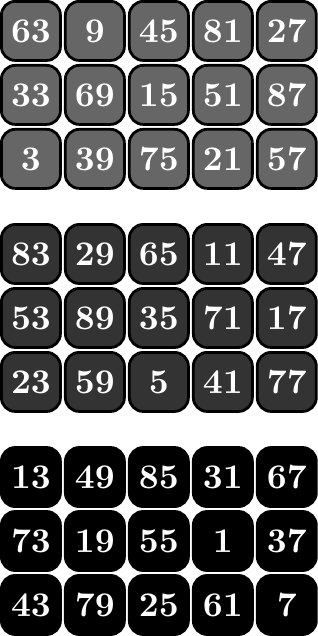}}\label{fig_QCCvertintermed}}\hfill
 \subfigure[All six $3 \times 5$ submatrices as in Fig.~\ref{fig_QCCvertintermed}. Here, the entry of the $\ith{i}$ row and the $\ith{j}$ column of $(m_{i3+g,j2+h})$ is the value $\mapb{i}{j} + g 3 +h(-5) \bmod 15$ that is equivalent to  $6(\map{i3+g}{j2+h}-g2-h3) \bmod 90$ according to~\eqref{eq_MappingToOneCodewordBack}. The entries are the coefficients of the six univariate polynomials $c_{0,0}(X)$, $c_{1,0}(X)$, $c_{2,0}(X)$ (left column) and $c_{0,1}(X)$, $c_{1,1}(X)$, $c_{2,1}(X)$ (right column) as in~\eqref{eq_UnivariateFinalQCCQCC}.]{\resizebox{.18\textwidth}{!}{\includegraphics{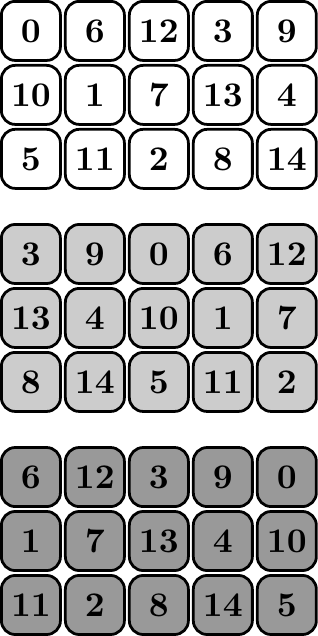}}\resizebox{.18\textwidth}{!}{\includegraphics{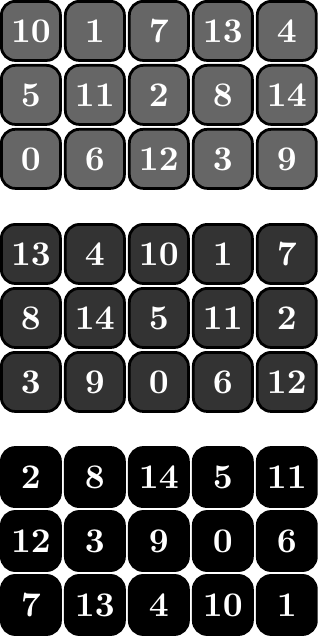}}\label{fig_QCCvert}
   }
\caption{Illustration of $\map{i}{j}$ as in~\eqref{def_MappingMatrixPolyQCCQCC} and $\mapb{i}{j}$ as in~\eqref{eq_MappingSubCodewordQCCQCC} for a $6$-quasi-cyclic product code $\QCCa \otimes \QCCb$ of length $6 \cdot 15$, where the row-code $\QCCa$ is $2$-quasi-cyclic and has length $2 \cdot 5$ and the column-code $\QCCb$ is $3$-quasi-cyclic and has length $3 \cdot 3$. The mapping $\map{i}{j}$ to one univariate polynomial is shown in Fig.~\ref{fig_QCCQCCfull}. The mapping $\mapb{i}{j}$ to six univariate polynomials is illustrated in Fig.~\ref{fig_QCCvertintermed} and Fig.~\ref{fig_QCCvert}.}
\label{fig_QCCQCC}
\end{figure}

We consider the entry in the $2^{\text{nd}}$ row and the $2^{\text{nd}}$ column of the full $9 \times 10$ matrix $(m_{i,j})_{i \in \interval{9}}^{j \in \interval{10}} \in \QCCa \otimes \QCCb$ shown in Fig.~\ref{fig_QCCQCCfull}. According to~\eqref{def_MappingMatrixPolyQCCQCC}, we have $\map{2}{2} = 76$, i.e., the coefficient of $X^{76}$ of the univariate polynomial $c(X) \in \QCCa \otimes \QCCb$ is $c_{76} = m_{2,2}$. The entry $m_{2,2}$ belongs to the $3 \times 5$ submatrix $(m_{i\QCCbl+2,j\QCCal})_{i \in \interval{3}}^{j \in \interval{5}}$ (bottom leftmost submatrix in Fig.~\ref{fig_QCCvertintermed} and in Fig.~\ref{fig_QCCvert}, with parameters $g=2$ and $h=0$). The entry in the $\ith{0}$ row and the $1^{\text{st}}$ column of the submatrix $(m_{i\QCCbl+2,j\QCCal})$ is the coefficient of $X^{12}$ of the polynomial $c_{2,0}(X)$, because $\submap{2}{0} + \mapb{0}{1} = 6 + 6 = 12$ according to~\eqref{eq_UnivariateFinalQCCQCC}. Via~\eqref{eq_ProofMatrixToUnivariateFirst}, it can be verified that the coefficient of $X^{12}$ of $c_{2,0}(X)$ is the coefficient of $X^{76}$ of $c(X) \in \QCCa \otimes \QCCb$.
\end{example}
In the following theorem, we state a basis of the $\QCCl$-quasi-cyclic product code $\QCCa \otimes \QCCb$ in terms of the two given generator matrices of $\QCCa$ and $\QCCb$ in RGB/POT form.
\begin{theorem}[Unreduced Basis of a Quasi-Cyclic Product Code] \label{theo_ProductCodeQCCQCCUnred}
Let $\QCCa$ be an $\LIN{\QCCal \cdot \QCCam}{\QCCak}{\QCCad}{q}$ $\QCCal$-quasi-cyclic code with generator matrix $\genmat[A] \in \Fx{q}^{\QCCal \times \QCCal}$ as in~\eqref{eq_GroebMatrixCodeA}, let $\QCCb$ be an $\LIN{\QCCbl \cdot \QCCbm}{\QCCbk}{\QCCbd}{q}$ $\QCCbl$-quasi-cyclic code with generator matrix $\genmat[B] \in \Fx{q}^{\QCCbl \times \QCCbl}$ as in~\eqref{eq_GroebMatrixCodeB}. Let $\QCCl = \QCCal \QCCbl$ and $\QCCm = \QCCam \QCCbm$.

Let $\mathbf{c}(X) = (c_{0,0}(X) \ c_{1,0}(X) \ \cdots \ c_{\QCCbl-1,0}(X)  \ \cdots \ c_{\QCCbl-1,\QCCal-1}(X) ) \in \Fx{q}^{\QCCl}$ be a codeword in $\QCCa \otimes \QCCb$, where $c_{g,h}(X), \forall g \in \interval{\QCCbl}, h \in \interval{\QCCal}$, is as defined in~\eqref{eq_UnivariateFinalQCCQCC}.

Then, a generator matrix in unreduced form with entries in $\Fx{q}$ of the $\QCCl$-quasi-cyclic product code $\QCCa \otimes \QCCb$ is given by
\begin{equation} \label{eq_GenMatrixQCCQCCUnreduced}
\mathbf{U}(X) = 
\begin{pmatrix}
\mathbf{U}^0(X) \\ 
\mathbf{U}^1(X)   
\end{pmatrix}, 
\end{equation}
where
\begin{equation} \label{eq_UnReducedBasisPart1}
\begin{split}
& \mathbf{U}^0(X)  = 
\begin{pmatrix}
\genunred[0][0] & \genunred[0][1]  & \cdots  & \cdots  & \genunred[0][\QCCl-1] \\
 & \genunred[1][1] & \cdots & \cdots & \genunred[1][\QCCl-1] \\
& & \ddots & \vdots & \vdots \\
\multicolumn{3}{c}{\bigzero} & \genunred[\QCCl-2][\QCCl-2] & \genunred[\QCCl-2][\QCCl-1] \\
 &  & & & \genunred[\QCCl-1][\QCCl-1]
\end{pmatrix} \\
& \qquad \qquad \qquad \cdot \diag \left( 1, X^{\submap{1}{0}}, \dots, X^{\submap{\QCCbl-1}{0}},\dots, X^{\submap{\QCCbl-1}{\QCCal-1}} \right),
\end{split}
\end{equation}
and where
\begin{equation} \label{eq_ElementsProductCodeQCCQCCUnred}
\begin{split}
\genunred[g+h\QCCbl][g^{\prime} + h^{\prime}\QCCbl] & = \genarg{A}[h][h^{\prime}]{X^{\intb \QCCbn}} \genarg{B}[g][g^{\prime}]{X^{\inta \QCCan}} \mod (X^{\QCCm}-1),  \\
& \qquad \qquad \qquad \forall g \in \interval{\QCCbl}, h \in \interval{\QCCal},g^{\prime} \in \interval{g, \QCCbl}, h^{\prime} \in \interval{h, \QCCal},
\end{split}
\end{equation}
and  
\begin{equation*} 
\mathbf{U}^1(X) = (X^{\QCCm}-1) \mathbf{I}_{\QCCl}, 
\end{equation*}
where $\mathbf{I}_{\QCCl}$ denotes the $\QCCl \times \QCCl$ identity matrix. The function $\submap{g}{h}$ is as defined in~\eqref{eq_ShiftCodeword}.
\end{theorem}
\begin{proof}
To get an explicit expression for the entries $\genunred[g+h\QCCbl][g^{\prime} + h^{\prime}\QCCbl]$ of the matrix $\mathbf{U}^0(X) \in \Fx{q}^{\QCCl \times \QCCl}$ as in~\eqref{eq_UnReducedBasisPart1}, we define a subcode of the product code $\QCCa \otimes \QCCb$ that is generated by one row of $\mathbf{U}^0(X)$ as given in~\eqref{eq_UnReducedBasisPart1}. 

Let $\QCCa^{(h)}$ denote a subcode of the given $\QCCal$-quasi-cyclic code $\QCCa$ that is spanned by 
\begin{align} \label{eq_SubCodeOfQCC}
\mathbf{a}^{(h)}(X) & \defeq  \big( 0 \ \cdots \ 0 \ \gen[A][h][h]  \ \gen[A][h][h+1] \ \cdots \ \gen[A][h][\QCCal-1]  \big),  \quad \forall h \in \interval{\QCCal}.
\end{align}
Similarly, let $\QCCb^{(g)}$ be a subcode of the given $\QCCbl$-quasi-cyclic code $\QCCb$ spanned by
\begin{align} \label{eq_SubCodeOfQCC}
\mathbf{b}^{(g)}(X) & \defeq  \big( 0 \ \cdots \ 0 \ \gen[B][g][g]  \ \gen[B][g][g+1] \ \cdots \ \gen[B][g][\QCCbl-1]  \big),  \quad \forall g \in \interval{\QCCbl}.
\end{align}
Clearly, we have for the row-code $\QCCa$ and the column-code $\QCCb$ that:
\begin{equation} \label{eq_QCCSumOfSubCodes}
\QCCa = \bigoplus_{h=0}^{\QCCal-1} \QCCa^{(h)}, \qquad \QCCb = \bigoplus_{g=0}^{\QCCbl-1} \QCCb^{(g)},
\end{equation}
and 
\begin{equation} \label{eq_ProductQCCSumOfSubCodes}
\QCCa \otimes \QCCb = \bigoplus_{g,h} \left( \QCCa^{(h)} \otimes \QCCb^{(g)} \right).
\end{equation}
Let $k_{A,h} = \QCCam - \deg \gen[A][h][h]$. The subcode $\QCCa^{(h)}$ is spanned by $\{X^{\alpha} \mathbf{a}^{(h)}(X): \alpha \in \interval{k_{A,h}}\}$. As in~\eqref{eq_VectorToUnivariatePoly}, a codeword of $\QCCa^{(h)}$ is an $\F{q}$-linear combination of $a^{(h,\alpha)}(X) \defeq \sum_{h^{\prime}=h}^{\QCCal-1} X^{h^{\prime}} a^{(\alpha)}_{h, h^{\prime}}(X^{\QCCal})$, where  $a^{(\alpha)}_{h, h^{\prime}}(X) = X^{\alpha} \gen[A][h][h^{\prime}]$. More explicitly if $g_{h,h'}^A (X) = \sum_{u=0}^{\QCCam-1} g_{h,h',u}^A X^u$, we have:
\begin{align} \label{eq_SubCodeACW2}
a^{(h,\alpha)}(X) & = \sum_{h^{\prime}=h}^{\QCCal-1} X^{h^{\prime}} X^{\QCCal \alpha} \genarg{A}[h][h^{\prime}]{X^{\QCCal}}  = \sum_{h^{\prime}=h}^{\QCCal-1} \sum_{u=0}^{\QCCam-1} \gencoeff{A}{h,h^{\prime},u} X^{h^{\prime} + \QCCal \alpha + \QCCal u },  \quad \forall h \in \interval{\QCCal}, \alpha \in \interval{k_{A,h}}.
\end{align}
Similarly, let $k_{B,g} = \QCCbm - \deg \gen[B][g][g]$. The subcode $\QCCb^{(g)}$ is spanned by $\{X^{\beta} \mathbf{b}^{(g)}(X): \beta \in \interval{k_{B,g}}\}$ and therefore if $g_{g,g'}^B(X) = \sum_{v=0}^{\QCCbm-1} g_{g,g',v}^B X^v$, a codeword of $\QCCb^{(g)}$ is an $\F{q}$-linear combination of 
\begin{align} \label{eq_SubCodeACW2}
b^{(g,\beta)}(X) & \defeq \sum_{g^{\prime}=g}^{\QCCbl-1} X^{g^{\prime}} X^{\QCCbl \beta} \genarg{B}[g][g^{\prime}]{X^{\QCCbl}} = \sum_{g^{\prime}=g}^{\QCCbl-1} \sum_{v=0}^{\QCCbm-1} \gencoeff{B}{g,g^{\prime},v} X^{g^{\prime} + \QCCbl \beta + \QCCbl v }, \quad \forall g \in \interval{\QCCbl}, \beta \in \interval{k_{B,g}}.
\end{align}
By definition of the product code $\QCCa \otimes \QCCb$ as in~\eqref{eq_ProductQCCSumOfSubCodes}, in the product array of $X^{\alpha} \mathbf{a}^{(h)}(X) \otimes X^{\beta} \mathbf{b}^{(g)}(X)$, the $\ith{(i,j)}$ entry is 
\begin{equation} \label{eq_SubcodeIndexes}
\gencoeff{A}{h,h^{\prime},u} \gencoeff{B}{g,g^{\prime},v}, 
\end{equation}
where 
\begin{equation} \label{eq_IndexesWithinMatrix}
\begin{split}
i & = g^{\prime} + \QCCbl \beta + \QCCbl v,\\
j & = h^{\prime} + \QCCal \alpha + \QCCal u.
\end{split}
\end{equation}
By Lemma~\ref{lem_MappingToUnivariatePolyQCC}, the corresponding codeword in $\QCCa \otimes \QCCb$ is then an $\F{q}$-linear combination of
\begin{equation} \label{eq_MatrixEntriesCoeff}
\sum_{h^{\prime},g^{\prime},u,v} \gencoeff{A}{h,h^{\prime},u} \gencoeff{B}{g,g^{\prime},v} X^{\map{i}{j}}.
\end{equation}
With $\map{i}{j}$ as in~\eqref{def_MappingMatrixPolyQCCQCC}, and with $i,j$ as in~\eqref{eq_IndexesWithinMatrix}, we obtain from~\eqref{eq_MatrixEntriesCoeff}:
\begin{align} 
\sum_{h^{\prime},g^{\prime},u,v} \gencoeff{A}{h,h^{\prime},u} \gencoeff{B}{g,g^{\prime},v} X^{\map{i}{j}} & = \sum_{h^{\prime},g^{\prime},u,v} \gencoeff{A}{h,h^{\prime},u} \gencoeff{B}{g,g^{\prime},v} X^{i \inta \QCCan \QCCal + j \intb \QCCbn \QCCbl} \nonumber \\
& = \sum_{h^{\prime},g^{\prime},u,v} \gencoeff{A}{h,h^{\prime},u} \gencoeff{B}{g,g^{\prime},v} X^{g^{\prime} \inta \QCCan \QCCal + \beta \inta \QCCan \QCCl + v \inta \QCCan \QCCl } X^{h^{\prime} \intb \QCCbn \QCCbl + \alpha \intb \QCCbn \QCCl + u \intb \QCCbn \QCCl} \nonumber  \\
& = X^{\QCCl (\beta \inta \QCCan +  \alpha \intb \QCCbn) } \sum_{h^{\prime},g^{\prime},u,v} \gencoeff{A}{h,h^{\prime},u} \gencoeff{B}{g,g^{\prime},v} X^{v \inta \QCCan \QCCl } X^{u \intb \QCCbn \QCCl} X^{g^{\prime} \inta \QCCan \QCCal} X^{h^{\prime} \intb \QCCbn \QCCbl} \nonumber \\
& = X^{\QCCl (\beta \inta \QCCan +  \alpha \intb \QCCbn) } \sum_{h^{\prime},g^{\prime}} \gencoeff{A}{h,h^{\prime}}(X^{\intb \QCCbn \QCCl}) \gencoeff{B}{g,g^{\prime}}(X^{\inta \QCCan \QCCl}) X^{g^{\prime} \inta \QCCan \QCCal} X^{h^{\prime} \intb \QCCbn \QCCbl}.  \label{eq_MatrixToPolyUnred}
\end{align}
With $\mapb{\beta}{\alpha}$ as defined in~\eqref{eq_MappingSubCodewordQCCQCC} and using~\eqref{eq_BEzoutRel}, we can reformulate~\eqref{eq_MatrixToPolyUnred} as follows:
\begin{align}
X^{\QCCl (\beta \inta \QCCan +  \alpha \intb \QCCbn) } & \sum_{h^{\prime},g^{\prime}} \gencoeff{A}{h,h^{\prime}}(X^{\intb \QCCbn \QCCl}) \gencoeff{B}{g,g^{\prime}}(X^{\inta \QCCan \QCCl}) X^{g^{\prime} \inta \QCCan \QCCal} X^{h^{\prime} \intb \QCCbn \QCCbl} \nonumber \\
& = X^{\QCCl \mapb{\beta}{\alpha}} \sum_{h^{\prime},g^{\prime}} \gencoeff{A}{h,h^{\prime}}(X^{\intb \QCCbn \QCCl}) \gencoeff{B}{g,g^{\prime}}(X^{\inta \QCCan \QCCl}) X^{g^{\prime} \QCCal + h^{\prime} \QCCbl} X^{\QCCal \QCCbl (- g^{\prime} \intb \QCCbm -h^{\prime} \inta \QCCam)} \nonumber \\
& = X^{\QCCl \mapb{\beta}{\alpha}} \sum_{h^{\prime},g^{\prime}} \gencoeff{A}{h,h^{\prime}}(X^{\intb \QCCbn \QCCl}) \gencoeff{B}{g,g^{\prime}}(X^{\inta \QCCan \QCCl}) X^{g^{\prime} \QCCal + h^{\prime} \QCCbl} X^{\QCCl \submap{g^{\prime}}{h^{\prime}}}, \label{eq_StripeOfCodeword} 
\end{align}
where $\submap{g^{\prime}}{h^{\prime}}$ is as in~\eqref{eq_ShiftCodeword}.
With~\eqref{eq_ProofMatrixToUnivariateFirst} of Lemma~\ref{lem_MappingUnivariateQCCQCC}, the $\ith{(g^{\prime},h^{\prime})}$ polynomial of the codeword $X^{\alpha} \mathbf{a}^{(h)}(X) \otimes X^{\beta} \mathbf{b}^{(g)}(X)$ in $\QCCa \otimes \QCCb$, in the form of a vector of $\QCCal \QCCbl$ univariate polynomials, is from~\eqref{eq_StripeOfCodeword}: 
\begin{equation}
X^{\mapb{\beta}{\alpha}} \gencoeff{A}{h,h^{\prime}}(X^{\intb \QCCbn}) \gencoeff{B}{g,g^{\prime}}(X^{\inta \QCCan}) X^{\submap{g^{\prime}}{h^{\prime}}} = X^{\mapb{\beta}{\alpha}} u_{g + h \QCCbl, g^{\prime} + h^{\prime}\QCCbl}(X) X^{\submap{g^{\prime}}{h^{\prime}}}.
\end{equation}
Hence $X^{\alpha} \mathbf{a}^{(h)}(X) \otimes X^{\beta} \mathbf{b}^{(g)}(X) $ is given by 
\begin{equation}
X^{\mapb{\beta}{\alpha}} \left(0 \ \cdots \ 0 \ u_{g + h \QCCbl, g + h\QCCbl}(X) \ \cdots \ u_{g + h \QCCbl, \QCCl-1}(X)\right) \diag \left( 1,\dots, X^{\submap{\QCCbl-1}{\QCCal-1}} \right), \quad \forall \alpha \in \interval{k_{A,h}}, \beta \in \interval{k_{B,g}},
\end{equation}
and therefore the subcode $\QCCa^{(h)} \otimes \QCCb^{(g)}$ is in the subspace generated by 
$\left(0 \ \cdots \ 0 \ u_{g + h \QCCbl, g + h\QCCbl}(X) \ \cdots \ u_{g + h \QCCbl, \QCCl-1}(X)\right) \cdot \diag \left( 1, X^{\submap{1}{0}}, \dots, X^{\submap{\QCCbl-1}{\QCCal-1}} \right)$.
Furthermore, we know that 
\begin{equation} \label{eq_BasisProductCode}
X^{\gamma} \left(0 \ \cdots \ 0 \ u_{g + h \QCCbl, g + h\QCCbl}(X) \ \cdots \ u_{g + h \QCCbl, \QCCl-1}(X)\right) \cdot \diag \left( 1, X^{\submap{1}{0}}, \dots, X^{\submap{\QCCbl-1}{\QCCal-1}} \right)
\end{equation}
equals $X^{\gamma} \mathbf{a}^{(h)}(X) \otimes X^{\gamma} \mathbf{b}^{(g)}(X) $, because $\mapb{\gamma}{\gamma} = \gamma \inta \QCCan + \gamma \intb \QCCbn = \gamma \bmod \QCCm$. Hence~\eqref{eq_BasisProductCode} spans $\QCCa^{(h)} \otimes \QCCb^{(g)}$ for all $\gamma$ and therefore the subcode $\QCCa^{(h)} \otimes \QCCb^{(g)}$ is the subspace generated by \eqref{eq_BasisProductCode}, for $\gamma \in [\QCCm)$.
\end{proof}
We consider the unreduced generating set of a $6$-quasi-cyclic product code in the following example according to Thm.~\ref{theo_ProductCodeQCCQCCUnred}.
\begin{example}[Unreduced Basis of a $6$-Quasi-Cyclic Product Code] \label{ex_QCCQCCProductb}
Let $\QCCa$ be a $2$-quasi-cyclic code of length $\QCCan = 2\QCCam$ and let $\QCCb$ be a $3$-quasi-cyclic code of length $\QCCbn = 3 \QCCbm$, where $\gcd(\QCCan, \QCCbn) = 1$. Let $\QCCm = \QCCam \QCCbm$. The generator matrices of $\QCCa$ and $\QCCb$ in RGB/POT form are
\begin{equation*}
\genmat[A] =
\begin{pmatrix}
\gen[A][0][0] & \gen[A][0][1]   \\
0 & \gen[A][1][1] 
\end{pmatrix}
\quad
\text{and}
\quad
\genmat[B] = 
\begin{pmatrix}
\gen[B][0][0] & \gen[B][0][1] & \gen[B][0][2] \\ 
0  & \gen[B][1][1] & \gen[B][1][2] \\
0 & 0 & \gen[B][2][2] 
\end{pmatrix}.
\end{equation*}
The unreduced basis $(\mathbf{U}^0(X) \ \mathbf{U}^1(X))^T$ of the $6$-quasi-cyclic product code $\QCCa \otimes \QCCb$ as in~\eqref{eq_GenMatrixQCCQCCUnreduced} is:
\begin{align} \label{eq_ExampleGenMatrix6QCC}
\mathbf{U}^0(X) & =
\begin{pmatrix}
\genunred[0][0] &  \genunred[0][1] & \genunred[0][2]  & \genunred[0][3] & \genunred[0][4] & \genunred[0][5] \\
 & \genunred[1][1] & \genunred[1][2] & 0 & \genunred[1][4] & \genunred[1][5]  \\
 &  &  \genunred[2][2] & 0 & 0 & \genunred[2][5]  \\
 & & & \genunred[3][3] & \genunred[3][4] & \genunred[3][5]  \\
\multicolumn{3}{c}{\bigzero} & & \genunred[4][4] &  \genunred[4][5]  \\
 & & & & & \genunred[5][5]
\end{pmatrix}\\
& \qquad \cdot \diag \left( 1, X^{\submap{1}{0}}, X^{\submap{2}{0}}, X^{\submap{0}{1}}, X^{\submap{1}{1}}, X^{\submap{2}{1}} \right), \nonumber
\end{align}
and $\mathbf{U}^1(X) = (X^{\QCCm}-1) \mathbf{I}_{6}$.
With $Y=X^{\intb \QCCbn}$, $Z=X^{\inta \QCCan}$, we can write~\eqref{eq_ExampleGenMatrix6QCC} explicitly 
\begin{align*}
\mathbf{U}^0(X) =
& \begin{pmatrix} 
\genarg{A}[0][0]{Y} \genarg{B}[0][0]{Z} & \genarg{A}[0][0]{Y} \genarg{B}[0][1]{Z} & \genarg{A}[0][0]{Y} \genarg{B}[0][2]{Z} & \genarg{A}[0][1]{Y} \genarg{B}[0][0]{Z} & \genarg{A}[0][1]{Y} \genarg{B}[0][1]{Z}  & \genarg{A}[0][1]{Y} \genarg{B}[0][2]{Z}  \\
0  & \genarg{A}[0][0]{Y} \genarg{B}[1][1]{Z}  & \genarg{A}[0][0]{Y} \genarg{B}[1][2]{Z} & 0  & \genarg{A}[0][1]{Y} \genarg{B}[1][1]{Z} & \genarg{A}[0][1]{Y} \genarg{B}[1][2]{Z} \\
0  & 0  & \genarg{A}[0][0]{Y} \genarg{B}[2][2]{Z} & 0 & 0  & \genarg{A}[0][1]{Y} \genarg{B}[2][2]{Z} \\
0  & 0 & 0 & \genarg{A}[1][1]{Y} \genarg{B}[0][0]{Z} & \genarg{A}[1][1]{Y}  \genarg{B}[0][1]{Z} & \genarg{A}[1][1]{Y} \genarg{B}[0][2]{Z} \\
0 & 0 & 0 & 0 & \genarg{A}[1][1]{Y} \genarg{B}[1][1]{Z} & \genarg{A}[1][1]{Y} \genarg{B}[1][2]{Z} \\
0 & 0 & 0 & 0 & 0 & \genarg{A}[1][1]{Y} \genarg{B}[2][2]{Z}
\end{pmatrix}\\
& \qquad \cdot \diag \left( 1, X^{\submap{1}{0}}, X^{\submap{2}{0}}, X^{\submap{0}{1}}, X^{\submap{1}{1}}, X^{\submap{2}{1}} \right),
\end{align*}
where each \nonzero{} non-diagonal element is taken modulo $(X^{\QCCm}-1)$.
Note that the matrix $\mathbf{U}^0(X)$ in~\eqref{eq_ExampleGenMatrix6QCC} has a zero entry at the same position as the Kronecker product of $\genmat[A] \otimes \genmat[B]$.
\end{example}
In the following, we derive a reduced basis of a $2$-quasi-cyclic product code $\QCCa \otimes \QCCb$, where $\QCCal=2$ and $\QCCbl=1$. As in Lemma~\ref{lem_EquivalenceSpectralAnalysis}, we denote the polynomials of the Pre-RGB/POT form that can be different from their counterparts in the RGB/POT form by a bar.
\begin{theorem}[Generator Matrix of a 2-Quasi-Cyclic Product Code in Pre-RGB/POT Form] \label{theo_ProductCode2-QCCTimesCyclic}
Let $\QCCa$ be an $\LIN{ \QCCan = 2 \cdot \QCCam}{\QCCak}{\QCCad}{q}$ $2$-quasi-cyclic code with generator matrix $\genmat[A] \in \Fx{q}^{2 \times 2}$ as in~\eqref{eq_GroebMatrixCodeA} and let $\QCCb$ be an $\LIN{\QCCbn = \QCCbm}{\QCCbk}{\QCCbd}{q}$ cyclic code with generator polynomial $\gen[B] \in \Fx{q}$. Let $\QCCm = \QCCam \QCCbm$. Then, a generator matrix in $\Fx{q}^{2 \times 2}$ in Pre-RGB/POT form as in~\eqref{eq_PRE-RGBPOTForm} of the $2$-quasi-cyclic product code $\QCCa \otimes \QCCb$ is given by:
\begin{equation} \label{eq_GenMatrix2-QCCTimesCyclic}
\bar{\mathbf{G}}(X) = 
\begin{pmatrix}
g_{0,0}(X) &  \bar{g}_{0,1}(X) \\
0 & g_{1,1}(X)  \\
\end{pmatrix} \cdot
\diag(1, X^{-\inta \QCCam}),
\end{equation}
where
\begin{align}
g_{0,0}(X) & = \gcd \left( X^{\QCCm}-1, \genarg{A}[0][0]{X^{\intb \QCCbn}} \genarg{B}{X^{\inta \QCCan}} \right), \nonumber \\
& = u_0(X)(X^{\QCCm}-1) + v_0(X) \genarg{A}[0][0]{X^{\intb \QCCbn}} \genarg{B}{X^{\inta \QCCan}} \label{eq_Bezout1},
\end{align}
for some polynomials $u_0(X), v_0(X) \in \Fx{q}$, and
\begin{align*}
g_{1,1}(X) & = \gcd \left( X^{\QCCm}-1, \genarg{A}[1][1]{X^{\intb \QCCbn}} \genarg{B}{X^{\inta \QCCan}} \right),\\
\bar{g}_{0,1}(X) & = v_0(X) \genarg{A}[0][1]{X^{\intb \QCCbn}} \genarg{B}{X^{\inta \QCCan}}.
\end{align*}
\end{theorem}
\begin{proof}
Let two polynomials $u_1(X), v_1(X) \in \Fx{q}$ be such that
\begin{align*}
g_{1,1}(X) & = \gcd \Big( X^{\QCCm}-1, g_{1,1}^{A}(X^{\intb \QCCbn})  g^{B}(X^{\inta \QCCan}) \Big)  \nonumber \\
& = u_1(X)(X^{\QCCm}-1) + v_1(X) g_{1,1}^{A}(X^{\intb \QCCbn}) g^{B}(X^{\inta \QCCan}). 
\end{align*}
Now, we transform the basis of the preimage directly. We denote a new Row $i$ by $\rowop{i}'$ and give the operation between two matrices. For ease of notation, we omit the term $\diag(1, X^{-\inta \QCCam})$.
From Thm.~\ref{theo_ProductCodeQCCQCCUnred}, we have:
\begin{align}
& \begin{pmatrix} 
g_{0,0}^{A}(X^{\intb \QCCbn})  g^{B}(X^{\inta \QCCan}) &  g_{0,1}^{A}(X^{\intb \QCCbn}) g^{B}(X^{\inta \QCCan}) \\
0  & g_{1,1}^{A}(X^{\intb \QCCbn}) g^{B}(X^{\inta \QCCan})\\
X^{\QCCm}-1 & 0 \\
0 & X^{\QCCm}-1
\end{pmatrix} \nonumber \\
& \quad \rowop{0}' \leftarrow \rowop{0} \cdot v_0+ \rowop{2} \cdot u_0 \nonumber \\
& \begin{pmatrix} 
g_{0,0}(X) & v_0(X) g_{0,1}^{A}(X^{\intb \QCCbn}) g^{B}(X^{\inta \QCCan}) \\
g_{0,0}^{A}(X^{\intb \QCCbn })  g^{B}(X^{\inta \QCCan}) &  g_{0,1}^{A}(X^{\intb \QCCbn}) g^{B}(X^{\inta \QCCan}) \\
0  & g_{1,1}^{A}(X^{ \intb \QCCbn}) g^{B}(X^{\inta \QCCan})\\
X^{\QCCm  }-1 & 0 \\
0 & X^{\QCCm}-1
\end{pmatrix} \nonumber \\
& \quad \rowop{1}' \leftarrow \rowop{1} - \frac{g_{0,0}^{A}(X^{\intb \QCCbn})  g^{B}(X^{\inta \QCCan})}{g_{0,0}(X)} \cdot \rowop{0} \nonumber \\
& \quad \rowop{3}' \leftarrow \rowop{3} - \frac{X^{\QCCm}-1}{g_{0,0}(X)} \cdot \rowop{0} \nonumber \\
& \begin{pmatrix} \label{eq_RowRed3}
g_{0,0}(X) & v_0(X) g_{0,1}^{A}(X^{\intb \QCCbn }) g^{B}(X^{\inta \QCCan}) \\
0  & g_{0,1}^{A}(X^{\intb \QCCbn}) g^{B}(X^{\inta \QCCan}) \left( 1 - v_0(X)\frac{g_{0,0}^{A}(X^{\intb \QCCbn})  g^{B}(X^{\inta \QCCan})}{g_{0,0}(X)} \right) \\
0  & g_{1,1}^{A}(X^{\intb \QCCbn}) g^{B}(X^{\inta \QCCan})\\
0  & - \frac{X^{\QCCm}-1}{g_{0,0}(X)} v_0(X) g_{0,1}^{A}(X^{\intb \QCCbn}) g^{B}(X^{\inta \QCCan}) \\
0 & X^{\QCCm}-1
\end{pmatrix},
\end{align}
and with~\eqref{eq_Bezout1}, we can reformulate Row $\rowop{1}$ of the matrix in~\eqref{eq_RowRed3} to 
\begin{align}
& \begin{pmatrix} \label{eq_RowRed4}
g_{0,0}(X) & v_0(X) g_{0,1}^{A}(X^{\intb \QCCbn }) g^{B}(X^{\inta \QCCan}) \\
0  & g_{0,1}^{A}(X^{\intb \QCCbn}) g^{B}(X^{\inta \QCCan}) u_0(X)\frac{X^{\QCCm}-1}{g_{0,0}(X)} \\
0  & g_{1,1}^{A}(X^{\intb \QCCbn}) g^{B}(X^{\inta \QCCan})\\
0  & - \frac{X^{\QCCm}-1}{g_{0,0}(X)} v_0(X) g_{0,1}^{A}(X^{\intb \QCCbn}) g^{B}(X^{\inta \QCCan}) \\
0 & X^{\QCCm}-1
\end{pmatrix}.
\end{align}
Using that $u_0(X)$ and $v_0(X)$ are relatively prime, we can merge $\rowop{1}$ and $\rowop{3}$ of the matrix in~\eqref{eq_RowRed4} to:
\begin{align}
& \begin{pmatrix} 
g_{0,0}(X) & v_0(X) g_{0,1}^{A}(X^{\intb \QCCbn}) g^{B}(X^{\inta \QCCan}) \\
0  & g_{1,1}^{A}(X^{\intb \QCCbn}) g^{B}(X^{\inta \QCCan})\\
0  & \frac{X^{\QCCm}-1}{g_{0,0}(X)} g_{0,1}^{A}(X^{\intb \QCCbn }) g^{B}(X^{\inta \QCCan}) \\
0 & X^{\QCCm}-1
\end{pmatrix} \nonumber \\
& \quad \text{Merge } \rowop{1} \text{ and } \rowop{3}  \text{, because } g_{1,1}(X) = \gcd \left( X^{\QCCm}-1, g_{1,1}^{A}(X^{\intb \QCCbn}) g^{B}(X^{\inta \QCCan})\right)   \nonumber \\
& \begin{pmatrix} \label{eq_RowRed6}
g_{0,0}(X) & v_0(X) g_{0,1}^{A}(X^{\intb \QCCbn}) g^{B}(X^{\inta \QCCan}) \\
0  & \frac{X^{\QCCm}-1}{g_{0,0}(X)} g_{0,1}^{A}(X^{\intb \QCCbn}) g^{B}(X^{\inta \QCCan}) \\
0 & g_{1,1}(X)
\end{pmatrix}.
\end{align}
In the last step, we show that $g_{1,1}(X) \mid \frac{X^{\QCCm}-1}{g_{0,0}(X)} g_{0,1}^{A}(X^{\intb \QCCbn }) g^{B}(X^{\inta \QCCan})$ and therefore Row $\rowop{1}$ of the matrix as in~\eqref{eq_RowRed6} can be deleted. 
From~\cite[Eq. (4)]{lally_algebraic_2001}, we know that for any generator matrix $\genmat[A] \in \Fx{q}^{2 \times 2}$ in RGB/POT form, there exists a matrix $\mathbf{A}(X) =  (a_{i,j}^A(X))^{j \in \interval{2}}_{i \in \interval{2}} \in \Fx{q}^{2 \times 2}$ with $a_{1,0}^A(X)= 0$ such that 
\begin{align} \label{eq_FromRGBPOTForA}
\mathbf{A}(X) \genmat[A] & = (X^{\QCCam}-1) \mathbf{I}_{2}.
\end{align}
We have 
\begin{align}
g_{1,1}(X) &  = \gcd \left( g_{1,1}^A(X^{\intb \QCCbn}) g^{B}(X^{\inta \QCCan}) ,X^{\QCCm}-1 \right) \nonumber \\
& = \lcm \left( \gcd \left( g_{1,1}^A(X^{\intb \QCCbn}),X^{\QCCm}-1 \right), \gcd \left( g^{B}(X^{\inta \QCCan}), X^{\QCCm}-1 \right) \right). \label{eq_FirstStepSubstitute}
\end{align}
From~\eqref{eq_FromRGBPOTForA}, we obtain $g_{1,1}^A(X^{\intb \QCCbn}) = -g_{0,1}^A(X^{\intb \QCCbn}) a_{0,0}^A(X^{\intb \QCCbn}) /a_{0,1}^A(X^{\intb \QCCbn}) $ and inserted in~\eqref{eq_FirstStepSubstitute} leads to:
\begin{align} \label{eq_ExpressionForG11}
g_{1,1}(X) & = \lcm \left( \gcd \left( \frac{g_{0,1}^A(X^{\intb \QCCbn}) a_{0,0}^A(X^{\intb \QCCbn}) }{a_{0,1}^A(X^{\intb \QCCbn})},X^{\QCCm}-1 \right) , \gcd \left( g^{B}(X^{\inta \QCCan}), X^{\QCCm}-1 \right) \right).
\end{align}
From~\eqref{eq_ExpressionForG11}, we can conclude that:
\begin{align}
g_{1,1}(X) & \mid \lcm \left( \gcd \left( g_{0,1}^A(X^{\intb \QCCbn}) a_{0,0}^A(X^{\intb \QCCbn}) ,X^{\QCCm}-1 \right) , \gcd \left( g^{B}(X^{\inta \QCCan}),  X^{\QCCm}-1  \right) \right), \nonumber 
\end{align}
which implies that
\begin{align}
g_{1,1}(X) & \mid g_{0,1}^A(X^{\intb \QCCbn}) \lcm \left( \gcd \left( a_{0,0}^A(X^{\intb \QCCbn}) ,X^{\QCCm}-1 \right) , \gcd \left( g^{B}(X^{\inta \QCCan}), X^{\QCCm}-1  \right) \right). \label{eq_BeforeSubsta00}
\end{align}
The polynomial $X^{\QCCam}-1$ has no repeated roots and we have $a_{0,0}^A(X) g_{0,0}^A(X) = X^{\QCCam}-1$. Clearly $a_{0,0}^A(X)$ and $g_{0,0}^A(X)$ are co-prime, i.e., $\exists u(X), v(X) \in \Fx{q}$, such that
\begin{align*}
u(X) a_{0,0}^A(X) + v(X) g_{0,0}^A(X) & =  1,
\end{align*}
implying that
\begin{align*}
u(X^{\intb \QCCbn}) a_{0,0}^A(X^{\intb \QCCbn}) + v(X^{\intb \QCCbn}) g_{0,0}^A(X^{\intb \QCCbn}) & =  1.
\end{align*}
Hence, the polynomials $a_{0,0}^A(X^{\intb \QCCbn})$ and $g_{0,0}^A(X^{\intb \QCCbn})$ are also relatively prime.
From $a_{0,0}^A(X^{\intb \QCCbn}) g_{0,0}^A(X^{\intb \QCCbn}) = X^{\intb \QCCm}-1 $, we can conclude that
\begin{align}
\gcd(a_{0,0}^A(X^{\intb \QCCbn}), X^{\QCCm}-1) \gcd(g_{0,0}^A(X^{\intb \QCCbn}), X^{\QCCm}-1) &  = \gcd(X^{\QCCm \intb}-1, X^{\QCCm}-1) = X^{\QCCm}-1. \nonumber 
\end{align}
Therefore
\begin{align}
\gcd(a_{0,0}^A(X^{\intb \QCCbn}), X^{\QCCm}-1) & = \frac{X^{\QCCm}-1}{\gcd(g_{0,0}^A(X^{\intb \QCCbn}), X^{\QCCm}-1)}. \label{eq_ExpressionForA00}
\end{align}
Inserting~\eqref{eq_ExpressionForA00} in~\eqref{eq_BeforeSubsta00} leads to:
\begin{align} \label{eq_BeforeExtendinga}
g_{1,1}(X) & \mid g_{0,1}^A(X^{\intb \QCCbn}) \lcm \left( \frac{X^{\QCCm}-1}{\gcd(g_{0,0}^A(X^{\intb \QCCbn}), X^{\QCCm}-1)} , \gcd \left( g^{B}(X^{\inta \QCCan}), X^{\QCCm}-1  \right) \right).
\end{align}
Note that
\begin{align}
& \gcd \left(g_{0,0}^A(X^{\intb \QCCbn}) g^B(X^{\inta \QCCan}), X^{\QCCm}-1 \right) \mid \gcd \left( g_{0,0}^A(X^{\intb \QCCbn}),  X^{\QCCm}-1 \right) \gcd\left(g^B(X^{\inta \QCCan}), X^{\QCCm}-1 \right),
\end{align}
which is equivalent to
\begin{align}
& f(X) \gcd \left( g_{0,0}^A(X^{\intb \QCCbn}) g^B(X^{\inta \QCCan}), X^{\QCCm}-1 \right) =  \gcd \left(g_{0,0}^A(X^{\intb \QCCbn}),  X^{\QCCm}-1 \right) \gcd \left(g^B(X^{\inta \QCCan}), X^{\QCCm}-1 \right), \label{eq_MultipleForGCD}
\end{align}
for some $f(X) \in \Fx{q}$. 
Extending the numerator and the denominator of the first lcm-term in~\eqref{eq_BeforeExtendinga} by $\gcd(g^B(X^{\inta \QCCan}), X^{\QCCm}-1)$ gives:
\begin{align} \label{eq_BeforeExtendingb}
g_{1,1}(X)  & \mid g_{0,1}^A(X^{\intb \QCCbn}) \lcm \left( \frac{(X^{\QCCm}-1)\gcd(g^B(X^{\inta \QCCan}), X^{\QCCm}-1) }{ f(X) \gcd \left( g_{0,0}^A(X^{\intb \QCCbn}) g^B(X^{\inta \QCCan}), X^{\QCCm}-1 \right)} , \gcd \left( g^{B}(X^{\inta \QCCan}), X^{\QCCm}-1  \right) \right).
\end{align}
Multiplying the RHS of~\eqref{eq_BeforeExtendingb} by a polynomial in $\Fx{q}$ does not change the divisibility. We multiply with $f(X)$ as in~\eqref{eq_MultipleForGCD} and obtain:
\begin{align} \label{eq FullExtension}
g_{1,1}(X) & \mid g_{0,1}^A(X^{\intb \QCCbn}) \lcm \left( \frac{(X^{\QCCm}-1)\gcd(g^B(X^{\inta \QCCan}), X^{\QCCm}-1)}{ f(X) \gcd \left( g_{0,0}^A(X^{\intb \QCCbn}) g^B(X^{\inta \QCCan}), X^{\QCCm}-1 \right)} f(X) , \gcd \left( g^{B}(X^{\inta \QCCan}), X^{\QCCm}-1  \right) \right).
\end{align}
We can extract the obtained factors from~\eqref{eq FullExtension} and get:
\begin{align*}
g_{1,1}(X) & \mid g_{0,1}^A(X^{\intb \QCCbn})  \frac{X^{\QCCm}-1}{g_{0,0}(X)} \lcm \left( \gcd( g^{B}(X^{\inta \QCCan}), X^{\QCCm}-1 ), \gcd \left( g^{B}(X^{\inta \QCCan}), X^{\QCCm}-1  \right) \right),\\
& = g_{0,1}^A(X^{\intb \QCCbn})  \frac{X^{\QCCm}-1}{g_{0,0}(X)} \gcd \left( g^{B}(X^{\inta \QCCan}), X^{\QCCm}-1  \right),\\
g_{1,1}(X) & \mid g_{0,1}^A(X^{\intb \QCCbn}) \frac{X^{\QCCm}-1}{g_{0,0}(X)} g^{B}(X^{\inta \QCCan}).
\end{align*}
Therefore we can delete Row $\rowop{1}$ of the matrix as in~\eqref{eq_RowRed6} and obtain the matrix in Pre-RGB/POT form as in~\eqref{eq_GenMatrix2-QCCTimesCyclic}, where we omitted the term $\diag(1, X^{-\inta \QCCam})$ during the proof.
\end{proof}
We consider an example of the generator matrix of a binary $2$-quasi-cyclic product code $\QCCa \otimes \QCCb$ in Pre-RGB/POT form, where the row-code $\QCCa$ is $2$-quasi-cyclic and the column-code $\QCCb$ is cyclic.
\begin{example}[Binary $2$-Quasi-Cyclic Product Code] \label{ex_BinaryQCCode}
Let $\alpha$ be a $21^{\text{st}}$ root of unity in $\F{2^{12}} \cong \Fx{2}/(X^{12} + X^7 + X^6 + X^5 + X^3 + X + 1)$.
Let $\QCCa$ be a binary $\LIN{2 \cdot 21}{17}{8}{2}$ $2$-quasi-cyclic code with generator matrix in RGB/POT form:
\begin{equation*}
\genmat[A] =
\begin{pmatrix}
\gen[A][0][0] & \gen[A][0][1]\\
0 & \gen[A][1][1] 
\end{pmatrix},
\end{equation*}
where
\begin{align*}
\gen[A][0][0]  & = \minpoly{\alpha} \cdot \minpoly{\alpha^3} \cdot \minpoly{\alpha^7},\\
\gen[A][0][1] & = \gen[A][0][0] \cdot (X^2+1), \\
\gen[A][1][1] & = \gen[A][0][0] \cdot \minpoly{\alpha^9},
\end{align*}
where the minimal polynomial $\minpoly{\alpha^i}$ was defined in~\eqref{eq_MinPoly}.
 The common roots $\alpha^i$ of $\gen[A][0][0], \gen[A][0][1]$ and $\gen[A][1][1]$, where $i \in \coset{1} \cup \coset{3} \cup \coset{7} = \{1,2,3,4,6,7,8,11,12,14,16\}$ are eigenvalues of $\genmat[A]$ with multiplicity two and the corresponding eigenvectors span the full space $\F{2^{12}}^2$ (see~\eqref{eq_cyclotomiccoset} for the definition of a cyclotomic coset \coset{i}).

Let $\beta$ be a $\ith{5}$ root of unity and let $\gen[B] = \minpoly{\beta^0} = X+1$ be the generator polynomial of the $\LIN{5}{4}{2}{2}$ cyclic code $\CYCb$. Let $\inta = 3$ and $\intb = -25$ be such that~\eqref{eq_BEzoutRel} holds. Let $\gamma = \alpha \beta$ and we have 
\begin{align*}
X^{105}-1 & = \prod_{i \in  \substack{\{0,1,3,5,7,9,11,13,15,\\17,21,25,35,45,49\}}} \minpoly{\gamma^i}.
\end{align*}
According to Thm.~\ref{theo_ProductCode2-QCCTimesCyclic}, the generator matrix in Pre-RGB/POT form as defined in~\eqref{eq_PRE-RGBPOTForm} of the $\LIN{2 \cdot 105}{68}{16}{2}$ 2-quasi-cyclic product code $\QCCa \otimes \QCCb$ is 
\begin{equation*}
\bar{\mathbf{G}}(X) =
\begin{pmatrix}
\gen[0][0] & \genbar[0][1]\\
0 & \gen[1][1] 
\end{pmatrix} \cdot \diag \left( 1, X^{-3 \cdot 21} \right),
\end{equation*}
where
\begin{align*}
\gen[0][0]  & = \prod_{i \in  \substack{\{0,1,3,5,7,9,11,\\15,21,25,35,45\}}} \minpoly{\gamma^i},\\
\genbar[0][1] & \equiv v_0(X) \genarg{A}[0][1]{X^{-25 \cdot 5}} \genarg{B}{X^{3 \cdot 42 }} \mod (X^{105}-1) \\
& \equiv (X^{39} + X^{38} + X^{36} + X^{35} + X^{32} + X^{30} + X^{25} + X^{24} + X^{22} + X^{20} + X^{18} + X^{17} + X^{12} + X^{11} + X^{10} + \\
& \qquad  X^{6} + X^{3} + X^{2} + 1)(X^{95} + X^{91} + X^{76} + X^{71} + X^{70} + X^{55} + X^{51} + X^{50} + X^{46} + X^{31} + X^{30} + X^{25} + \\
& \qquad X^{21} + X^{11} + X^{10} + 1) \mod (X^{105}-1) \\
& \equiv  X^{95} + X^{92} + X^{91} + X^{90} + X^{89} + X^{86} + X^{85} + X^{84} + X^{82} + X^{80} + X^{75} + X^{72} + X^{71} + X^{69} + X^{67} + \\ 
& \qquad X^{62} + X^{61} +  X^{59} + X^{57} + X^{52} + X^{51} + X^{49} + X^{47} + X^{45} + X^{30} + X^{27} + X^{26} + X^{24} + X^{22} + X^{20} + \\
& \qquad  X^{15} + X^{12} + X^{11} + X^{10} + X^{9} + X^{6} + X^{5} + X^{4} + X^{2} + 1 \mod (X^{105}-1),\\
\gen[1][1] & = \gen[0][0] \cdot \minpoly{\gamma^9},
\end{align*}
where $\deg \gen[1][1] = 77 $. Performing row-reduction on $\bar{\mathbf{G}}(X)$ leads to the RGB/POT form, where:
\begin{align*}
\gen[0][1] & = X^{75} + X^{72} + X^{71} + X^{69} + X^{67} + X^{62} + X^{61} + X^{59} + X^{57} + X^{55} + X^{40} + X^{37} + X^{36} + X^{35} + X^{34} + \\
& \qquad X^{31} + X^{30} + X^{29} + X^{27} + X^{25} + X^{20} + X^{17} + X^{16} + X^{14} + X^{12} + X^{10}.
\end{align*}
\end{example}
The following theorem gives the generator matrix in RGB/POT form (as defined in~\eqref{def_GroebBasisMatrix}) of an $\QCCl$-quasi-cyclic product code $\QCCa \otimes \QCCb$, where the row-code $\QCCa$ is a $1$-level $\QCCl$-quasi-cyclic code and $\QCCb$ is a cyclic code (see Definition~\ref{def_LevelQC} for the property $1$-level). 
\begin{theorem}[Generator Matrix of a $1$-Level Quasi-Cyclic Product Code in RGB/POT Form] \label{theo_OneLevelQC}
Let $\QCCa$ be an $\LIN{\QCCan = \QCCl \cdot \QCCam}{\QCCak}{\QCCad}{q}$ $1$-level $\QCCl$-quasi-cyclic code with generator matrix in RGB/POT form:
\begin{align*}
\genmat[A]  & = \begin{pmatrix} \gen[A][0][0] & \hspace{.2cm} \gen[A][0][1] &  \hspace{.55cm} \cdots & \hspace{.3cm} \gen[A][0][\QCCl-1] 
\end{pmatrix} \nonumber \\
& = \begin{pmatrix}
\gen[A] & \gen[A] f_{1}^{A}(X) & \cdots & \gen[A] f_{\QCCl-1}^{A}(X)
\end{pmatrix} \label{eq_GenMatrixQCCOneLevel}
\end{align*}
as shown in Corollary~\ref{cor_OneLevelQC}. Let $\QCCb$ be an $\LIN{\QCCbn=\QCCbm}{\QCCbk}{\QCCbd}{q}$ cyclic code with generator polynomial $\gen[B] \in \Fx{q}$. Let $\QCCm = \QCCam \QCCbm$. Then the generator matrix of the $1$-level $\QCCl$-quasi-cyclic product code $\QCCa \otimes \QCCb$ in RGB/POT form is:
\begin{equation*}
\genmat =  
\begin{pmatrix}
\gen & \gen f_{1}^{A}(X^{\intb \QCCbn})  & \cdots & \gen f_{\QCCl-1}^{A}(X^{\intb \QCCbn })
\end{pmatrix} 
\cdot \diag \big( 1, X^{-\inta \QCCam}, X^{-2 \inta \QCCam}, \dots, X^{- (\QCCl-1) \inta \QCCam} \big),
\end{equation*}
where 
\begin{equation*} \label{eq_GCDOneLevel}
\gen = \gcd \left( X^{\QCCm}-1, \genarg{A}{X^{\intb \QCCbn}} \genarg{B}{X^{\inta \QCCan}} \right) .
\end{equation*}
\end{theorem}
\begin{proof}
Let two polynomials $u(X), v(X) \in \Fx{q}$ be such that:
\begin{equation} \label{eq_BezoutProductDiagOneLevel}
\begin{split}
\gen & = u(X) (X^{\QCCm}-1) + v(X) \genarg{A}{X^{\intb \QCCbn}} \genarg{B}{X^{\inta \QCCan }}.
\end{split}
\end{equation}
We show how to reduce the basis representation to the RGB/POT form. As in the proof of Thm.~\ref{theo_ProductCode2-QCCTimesCyclic}, we denote a new Row $i$ by $\rowop{i}'$. For ease of notation, we omit the term $\diag(1, X^{-\inta \QCCam}, X^{-2 \inta \QCCam}, \dots ,X^{- (\QCCl-1) \inta \QCCam} )$.

According to Thm.~\ref{theo_ProductCodeQCCQCCUnred}, the unreduced basis of $\QCCa \otimes \QCCb$ is:
\begingroup
\begin{align}
& \begin{pmatrix}
\genarg{A}{X^{\intb \QCCbn}} \genarg{B}{X^{\inta \QCCan}} & \genarg{A}{X^{\intb \QCCbn}} f_{1}^{A}(X^{\intb \QCCbn}) \genarg{B}{X^{\inta \QCCan}} & \cdots & \genarg{A}{X^{\intb \QCCbn}} f_{\QCCl-1}^{A}(X^{\intb \QCCbn}) \genarg{B}{X^{\inta \QCCan}} \\
X^{\QCCm}-1 &   \\
& X^{\QCCm}-1 & \multicolumn{2}{c}{\bigzero} \\
\multicolumn{2}{c}{\bigzero} & \ddots & \\
 &  & & X^{\QCCm}-1 
\end{pmatrix} \label{eq_StartMatrix}  \\[2ex]
& \rowop{0}' \leftarrow v(X)\rowop{0} + u(X) \rowop{1} + u(X) f_{1}^{A}(X^{\intb \QCCbn}) \rowop{2} + \dots + u(X) f_{\QCCl-1}^{A}(X^{\intb \QCCbn}) \rowop{\QCCl}  \nonumber \\[2ex]
& \begin{pmatrix}
\gen & \gen f_{1}^{A}(X^{\intb \QCCbn}) & \cdots & \gen f_{\QCCl-1}^{A}(X^{\intb \QCCbn}) \\
\genarg{A}{X^{\intb \QCCbn}} \genarg{B}{X^{\inta \QCCan}} & \genarg{A}{X^{\intb \QCCbn}} f_{1}^{A}(X^{\intb \QCCbn}) \genarg{B}{X^{\inta \QCCan}} & \cdots &  \genarg{A}{X^{\intb \QCCbn}} f_{\QCCl-1}^{A}(X^{\intb \QCCbn}) \genarg{B}{X^{\inta \QCCan}} \\
X^{\QCCm}-1 & \\
& X^{\QCCm }-1 &  \multicolumn{2}{c}{\bigzero} \\
\multicolumn{2}{c}{\bigzero} & \ddots & \\
&  & & X^{\QCCm}-1 
\end{pmatrix}, \label{eq_MatrixFirstMerge}
\end{align}
where the $\ith{i}$ entry in the new Row $\rowop{0}$ in the matrix in~\eqref{eq_MatrixFirstMerge} from matrix in~\eqref{eq_StartMatrix} was obtained using:
\begin{align}
v(X) \genarg{A}{X^{\intb \QCCbn}} f_{i}^{A}(X^{\intb \QCCbn}) \genarg{B}{X^{\inta \QCCan}} + & u(X) f_{i}^{A}(X^{\intb \QCCbn}) (X^{\QCCm} -1) \nonumber \\
& = f_{i}^{A}(X^{\intb \QCCbn}) \big( v(X) \genarg{A}{X^{\intb \QCCbn}} \genarg{B}{X^{\inta \QCCan}} + u(X) (X^{\QCCm} - 1) \big). \label{eq_PreGCDForm}
\end{align}
Inserting~\eqref{eq_BezoutProductDiagOneLevel} into~\eqref{eq_PreGCDForm} gives:
\begin{align*}
f_{i}^{A}(X^{\intb \QCCbn}) \big( v(X) \genarg{A}{X^{\intb \QCCbn}} \genarg{B}{X^{\inta \QCCan}} + u(X) (X^{\QCCm}-1) \big) = f_{i}^{A}(X^{\intb \QCCbn}) \gen.
\end{align*}
Clearly, $\gen$ divides $\genarg{A}{X^{\intb \QCCbn}} \genarg{B}{X^{\inta \QCCan}}$ and it is easy to check that Row $\rowop{1}$ of the matrix in~\eqref{eq_MatrixFirstMerge} can be obtained from Row $\rowop{0}$ by multiplying by $\genarg{A}{X^{\intb \QCCbn}} \genarg{B}{X^{\inta \QCCan}}/\gen$.
Therefore, we can omit the linearly dependent Row $\rowop{1}$ in~\eqref{eq_MatrixFirstMerge} and the reduced basis in RGB/POT form is:
\begin{align*}
& \begin{pmatrix}
\gen \hspace*{.3cm}  & \gen f_{1}^{A}(X^{\intb \QCCbn}) & \cdots   & \gen f_{\QCCl-1}^{A}(X^{\intb \QCCbn})
\end{pmatrix}, 
\end{align*}
where we omitted the matrix $\diag (1, X^{-\inta \QCCam}, X^{-2 \inta \QCCam}, \dots, X^{- (\QCCl-1) \inta \QCCam})$ during the proof.
\endgroup
\end{proof}
We conjecture the (general form of the) generator matrix in Pre-RGB/POT form of an $\QCCal \QCCbl$-quasi-cyclic product code $\QCCa \otimes \QCCb$ in the following. We reduced the unreduced basis of several examples and could verify Conjecture~\ref{conj_ProductCodeQCCQCC}. 
\begin{conjecture}[Generator Matrix of an $\QCCal \QCCbl$-Quasi-Cyclic Product Code in Pre-RGB/POT Form] \label{conj_ProductCodeQCCQCC}
Let $\QCCa$ be an $\LIN{\QCCan = \QCCal \cdot \QCCam}{\QCCak}{\QCCad}{q}$ $\QCCal$-quasi-cyclic code with generator matrix $\genmat[A] \in \Fx{q}^{\QCCal \times \QCCal}$ as in~\eqref{eq_GroebMatrixCodeA} and let $\QCCb$ be an $\LIN{\QCCbn = \QCCbl \cdot \QCCbm}{\QCCbk}{\QCCbd}{q}$ $\QCCbl$-quasi-cyclic code with generator matrix $\genmat[B] \in \Fx{q}^{\QCCbl \times \QCCbl}$ as in~\eqref{eq_GroebMatrixCodeB}. Let $\QCCm = \QCCam \QCCbm$ and $\QCCl = \QCCal \QCCbl$. 

Then, a generator matrix in $\Fx{q}^{\QCCl \times \QCCl}$ in Pre-RGB/POT form of the $\LIN{\QCCn = \QCCl \cdot \QCCm}{\QCCak \QCCbk}{\QCCad \QCCbd}{q}$ $\QCCl$-quasi-cyclic product code $\QCCa \otimes \QCCb$ is given by:
\begin{equation*} \label{eq_GenMatrixQCCQCC}
\begin{split}
\bar{\mathbf{G}}(X) & = 
\begin{pmatrix}
g_{0,0}(X) &  \bar{g}_{0,1}(X) & \cdots  & \cdots  & \bar{g}_{0,\QCCl-1}(X) \\
 & g_{1,1}(X) & \cdots & \cdots & \bar{g}_{1,\QCCl-1}(X) \\
& & \ddots & \vdots & \vdots \\
\multicolumn{3}{c}{\bigzero} & g_{\QCCl-2,\QCCl-2}(X) & \bar{g}_{\QCCl-2,\QCCl-1}(X) \\
 &  & & & g_{\QCCl-1,\QCCl-1}(X)
\end{pmatrix} \\
& \qquad \qquad \cdot \diag \left( 1, X^{\submap{1}{0}}, \dots, X^{\submap{\QCCbl-1}{0}},\dots, X^{\submap{\QCCbl-1}{\QCCal-1}} \right),
\end{split}
\end{equation*}
where the $\QCCl$ diagonal entries are
\begin{align} 
g_{g+h\QCCbl, g+h\QCCbl}(X) & = \gcd \Big( X^{\QCCm}-1, \genarg{A}[h][h]{X^{\intb \QCCbn}} \genarg{B}[g][g]{X^{\inta \QCCan}} \Big), \quad \forall g \in \interval{\QCCbl}, \forall  h \in \interval{\QCCal} \label{eq_DiagonalElementsProductCodeQCCQCC}.
\end{align}
Let the polynomials $u_{g,h}(X), v_{g,h}(X) \in \Fx{q}$ be such that:
\begin{equation*} \label{eq_BezoutProductDiag}
\begin{split}
\gen[g+h\QCCbl][g+h\QCCbl] & = u_{g,h}(X)  (X^{\QCCm}-1) + v_{g,h}(X) \genarg{A}[h][h]{X^{\intb \QCCbn}} \genarg{B}[g][g]{X^{\inta  \QCCan }}, \quad \forall g \in \interval{\QCCbl}, \forall  h \in \interval{\QCCal}.
\end{split}
\end{equation*}
Then the off-diagonal entries of the matrix $\bar{\mathbf{G}}(X)$ are given by
\begin{equation*}  \label{eq_NonDiagonalElementsProductCodeQCCQCC}
\begin{split}
\bar{g}_{g+h\QCCbl,g^{\prime} + h^{\prime}\QCCbl}(X) & = v_{g,h}(X) \genarg{A}[h][h^{\prime}]{X^{\intb \QCCbn}} \genarg{B}[g][g^{\prime}]{X^{\inta \QCCan}} \mod (X^{\QCCm}-1), \\
 & \qquad  \qquad \forall g \in \interval{\QCCbl}, h \in \interval{\QCCal},g^{\prime} \in \interval{g+1, \QCCbl}, h^{\prime} \in \interval{h+1, \QCCal}.
\end{split}
\end{equation*}
\end{conjecture}
Note that the expression of the diagonal terms in~\eqref{eq_DiagonalElementsProductCodeQCCQCC} is equivalent to the generator polynomial of a cyclic product code $\QCCa \otimes \QCCb$ where the cyclic row-code $\QCCa$ is generated by $\genarg{A}[h][h]{X}$ and the generator polynomial of the cyclic column-code $\QCCb$ is $\genarg{B}[g][g]{X}$.

\section{Spectral Analysis of a Quasi-Cyclic Product Code and Bounding The Minimum Hamming Distance} \label{sec_SpectralAnalysis}

\subsection{Spectral Analysis} \label{subsec_SpectralAnalysisProduct}
In this section, we apply the spectral techniques of Semenov and Trifonov~\cite{semenov_spectral_2012} to an $\QCCl$-quasi-cyclic product code $\QCCa \otimes \QCCb$, where $\QCCa$ is an $\QCCl$-quasi-cyclic code and $\QCCb$ is a cyclic code, and generalize the results for a cyclic product code as in~\cite[Thm. 4]{lin_further_1970}.
Furthermore, we bound the minimum Hamming distance of a given $\QCCl$-quasi-cyclic code $\QCCa$ by embedding it into an $\QCCl$-quasi-cyclic product code $\QCCa \otimes \QCCb$. This method extends the approach of~\cite[Thm. 4]{zeh_generalizing_2013}, where a lower bound on the minimum Hamming distance of a given cyclic code was obtained through embedding it into a cyclic product code.

It turns out that the eigenvalues of maximal multiplicity $\QCCl$ of the $\QCCl$-quasi-cyclic code $\QCCa$ and the zeros of $\QCCb$ occur in the spectral analysis of the $\QCCl$-quasi-cyclic code $\QCCa \otimes \CYCb$ with maximal multiplicity $\QCCl$. This is similar to the appearance of the zeros of two cyclic codes in the generator polynomial of their cyclic product code. The eigenvalues of multiplicity smaller than $\QCCl$ of $\QCCa$ and the \nonzero{}s of $\QCCb$ are reflected in the spectral analysis of $\QCCa \otimes \QCCb$ in a manner similar to that of the \nonzero{}s of two cyclic codes $\QCCa$ and $\QCCb$ in the case of a cyclic product code $\QCCa \otimes \QCCb$. Therefore, they are treated separately in Lemma~\ref{lem_EigenvaluesMaxMulti} and in Lemma~\ref{lem_EigenvaluesSmallerThanMax}.

Throughout this section, let $\QCCa$ be an $\LIN{\QCCan = \QCCl \cdot \QCCam}{\QCCak}{\QCCad}{q}$ $\QCCl$-quasi-cyclic code with generator matrix in RGB/POT form as in~\eqref{eq_GroebMatrixCodeA} and let $\CYCb$ be an $\LIN{\QCCbn = \QCCbm}{\QCCbk}{\QCCbd}{q}$ cyclic code with generator polynomial $\gen[B]$. Let $\QCCm = \QCCam \QCCbm$. The product code $\QCCa \otimes \QCCb$ is an $\LIN{\QCCl \cdot \QCCm}{\QCCak \QCCbk}{\QCCad \QCCbd}{q}$ $\QCCl$-quasi-cyclic code with generator matrix in Pre-RGB/POT form as given in Conjecture~\ref{conj_ProductCodeQCCQCC}, i.e., their entries are:
\begin{align} 
\gen[h][h] & = u_h(X) (X^{\QCCm}-1) + v_h(X) \genarg{A}[h][h]{X^{\intb \QCCbn}} \genarg{B}{X^{\inta \QCCan}}, & \forall h \in \interval{\QCCl} \label{eq_QCTimeCyclicDiag},\\
\bar{g}_{h,h^{\prime}}(X) & = v_h(X) \genarg{A}[h][h^{\prime}]{X^{\intb \QCCbn}} \genarg{B}{X^{\inta \QCCan}}, & \forall h \in \interval{\QCCl}, h^{\prime} \in \interval{h+1, \QCCl}. \label{eq_QCTimeCyclicNonDiag}
\end{align}
Furthermore, as in~\eqref{eq_BEzoutRel} let throughout this section two \nonzero{} integers $\inta, \intb$ be such that $\inta \QCCan + \intb \QCCbn = 1$.
For a given set $A = \{a_0, a_1, \dots, a_{|A|-1}\}$, denote by $\shiftset{A}{z} \defeq \{a_i+z  \mid  a_{i} \in A \}$.
\begin{lemma}[Eigenvalues Of Maximal Multiplicity] \label{lem_EigenvaluesMaxMulti}
Let $\QCCa$ be an $\LIN{\QCCl \cdot \QCCam}{\QCCak}{\QCCad}{q}$ $\QCCl$-quasi-cyclic code with generator matrix $\genmat[A]$ in RGB/POT form. Let $\alpha$ be an element of order $\QCCam$ in $\F{q^{\QCCas}}$, $\QCCb$ an $\LIN{\QCCbn = \QCCbm}{\QCCbk}{\QCCbd}{q}$ cyclic code, and $\beta$ an element of order $\QCCbm$ in $\F{q^{\QCCbs}}$. Define $\QCCs \defeq \lcm(\QCCas, \QCCbs)$. Let $\gamma \defeq \alpha \beta$ be in $\F{q^{\QCCs}}$.
Let the set $A^{(\QCCl)} \subseteq \interval{\QCCam}$ contain the exponents of all eigenvalues $\eigenvalue{z}^A = \alpha^{z}, \forall z \in A^{(\QCCl)}$ of $\QCCa$ of (algebraic and geometric) multiplicity $\QCCl$. Let $B \subseteq \interval{\QCCbm}$ be the defining set of $\QCCb$, i.e., the set of exponents of all roots of the generator polynomial $\gen[B] = \prod_{i \in B} (X-\beta^i)$ of $\QCCb$. Then, the set:
\begin{equation*}
C^{(\QCCl)} = A^{(\QCCl)} \cup \shiftset{A^{(\QCCl)}}{\QCCam} \cup \shiftset{A^{(\QCCl)}}{2\QCCam} \cup \cdots \cup 
\shiftset{A^{(\QCCl)}}{(\QCCbm-1)\QCCam} \cup
B \cup \shiftset{B}{\QCCbm} \cup \shiftset{B}{2\QCCbm} \cup \cdots \cup \shiftset{B}{(\QCCam-1)\QCCbm} 
\end{equation*}
is the set of all the exponents of the eigenvalues $\eigenvalue{z} = \gamma^{z}$ for all $z \in C^{(\QCCl)}$ of the $\QCCl$-quasi-cyclic product code $\QCCa \otimes \QCCb$ of maximal multiplicity $\QCCl$.
Furthermore, we have $|C^{(\QCCl)}| = |A^{(\QCCl)}| \QCCbm+ (\QCCbm - \QCCbk)\QCCam - |A^{(\QCCl)}| (\QCCbm - \QCCbk) = (\QCCbm - \QCCbk)\QCCam + |A^{(\QCCl)}| \QCCbk$. 
\end{lemma}
\begin{proof}
For an eigenvalue $\eigenvalue{z}^A = \alpha^{z}, \forall z \in A^{(\QCCl)}$ of the $\QCCl$-quasi-cyclic code $\QCCa$ of multiplicity $\QCCl$, all $\ell$ diagonal entries $\gen[A][h][h]$ of $\genmat[A]$ are divisible by $(X-\alpha^{z})$. 
From Conjecture~\ref{conj_ProductCodeQCCQCC}, we can conclude that if $\alpha^{z}$ is a root of $\gen[A][h][h]$, then 
\begin{equation} \label{eq_RootsMultiplication}
\gamma^{z}, \gamma^{z+ \QCCam}, \gamma^{z + 2\QCCam}, \dots, \gamma^{z + (\QCCbm-1) \QCCam}
\end{equation}
are $\QCCbm$ roots of $\genarg{A}[h][h]{X^{\intb \QCCbn}}$ and therefore of $\gen[][h][h]$ as in~\eqref{eq_QCTimeCyclicDiag}, because:
\begin{equation} \label{eq_MappingFromComponentToProductPre}
(\gamma^{z + i \QCCam})^{\intb \QCCbn}  = \gamma^{z \intb \QCCbn} =  \alpha^{z \intb \QCCbn} \beta^{z \intb \QCCbn} = \alpha^{z \intb \QCCbn}, 
\end{equation}
where in the first step we used the fact that the order of $\gamma$ is $\QCCam \QCCbn$.
The order of $\alpha$ is $\QCCam$ and with~\eqref{eq_BEzoutRel}, we obtain from~\eqref{eq_MappingFromComponentToProductPre}:
\begin{equation*} 
\alpha^{z \intb \QCCbn} = \alpha^{z \intb \QCCbn} \alpha^{z \inta \QCCan} = \alpha^z.
\end{equation*}
The zeros of the generator polynomial $\gen[B]$ of $\CYCb$ appear in the spectral analysis of the product code $\QCCa \otimes \QCCb$ similar to the eigenvalues of $\QCCa$ with multiplicity $\QCCl$. A \nonzero{} polynomial $\gen[][h][h], \forall h \in \interval{\QCCl}$ as given in~\eqref{eq_QCTimeCyclicDiag} has a zero at
\begin{equation*}
\gamma^{z}, \gamma^{z+\QCCbm}, \gamma^{z+2\QCCbm}, \dots, \gamma^{z+(\QCCam-1)\QCCbm}, 
\end{equation*}
if $\beta^{z}$ is a zero of $\gen[B]$. From~\cite[Thm. 3]{lin_further_1970}, we know that the polynomial $\gen[A][h][h]$ as in~\eqref{eq_QCTimeCyclicDiag} has a zero if and only if either the polynomial $\genarg{A}[h][h]{X^{\intb \QCCbn}}$ has a zero or the polynomial $\genarg{B}{X^{\inta \QCCan}}$ has a zero or both. Therefore, the polynomial $\prod_{h=0}^{\QCCl-1} \gen[][h][h]$ has a zero of multiplicity $\QCCl$ if and only if the polynomial $\prod_{h=0}^{\QCCl-1} \genarg{A}[h][h]{X^{\intb \QCCbn}}$ has a zero of multiplicity $\QCCl$ or the polynomial $\genarg{B}{X^{\inta \QCCan}}$ has a zero or both. The cardinality $|C^{(\QCCl)}|$ follows.
\end{proof}
The following lemma considers eigenvalues of the $\QCCl$-quasi-cyclic product code $\QCCa \otimes \QCCb$ of multiplicity smaller than $\QCCl$ and their corresponding eigenvectors. Lemma~\ref{lem_EigenvaluesSmallerThanMax} applies also to eigenvalues of $\QCCa \otimes \QCCb$ of multiplicity $r=0$, which are non-eigenvalues.
\begin{lemma}[Eigenvalues Of Smaller Multiplicity and Their Eigenvectors] \label{lem_EigenvaluesSmallerThanMax}
Let the two codes $\QCCa$ and $\QCCb$ with parameters be given as in Lemma~\ref{lem_EigenvaluesMaxMulti}.

Let the set $A^{(r)} \subseteq \interval{\QCCam}$ contain the exponents of all eigenvalues $\eigenvalue{z}^{A} = \alpha^{z}, \forall z \in A^{(r)}$ of $\QCCa$ of (algebraic and geometric) multiplicity $r \in \interval{\QCCl}$. Let $\eigenvector^{A}_{z, 0}, \eigenvector^{A}_{z, 1}, \dots, \eigenvector^{A}_{z, r-1} \in \F{q^{\QCCas}}^{\QCCl}$ be the corresponding $r$ eigenvectors of $\eigenvalue{z}^{A}$ as defined in~\eqref{eq_eigenvectors}, i.e., a basis of the right kernel of $\genmat[A][\eigenvalue{z}^{A}]$.
Let $B \subseteq \interval{\QCCbm}$ be the defining set of $\QCCb$, i.e., the set of exponents of all roots of the generator polynomial $\gen[B] = \prod_{i \in B} (X-\beta^i)$ of $\QCCb$.
Let $\gamma \defeq \alpha \beta$ be in $\F{q^{\QCCs}}$. Then, the set:
\begin{equation*}
C^{(r)} = \left( A^{(r)} \cup \shiftset{A^{(r)}}{\QCCam} \cup \shiftset{A^{(r)}}{2\QCCam} \cup \cdots \cup \shiftset{A^{(r)}}{(\QCCbm-1)\QCCam} \right) \setminus \left( B \cup \shiftset{B}{\QCCbm} \cup \shiftset{B}{2\QCCbm} \cup \cdots \cup \shiftset{B}{(\QCCam-1)\QCCbm} \right)
\end{equation*}
is the set of all exponents of the eigenvalues $\lambda_{z} = \gamma^{z}$ for all $z \in C^{(r)}$ of the $\QCCl$-quasi-cyclic product code $\QCCa \otimes \QCCb$ of multiplicity $r$. 
The number of eigenvalues of $\QCCa \otimes \QCCb$ of multiplicity $r$ is $|C^{(r)}| = |A^{(r)}| \QCCbk$.
Furthermore, the corresponding eigenvectors $\eigenvector_{z, 0}, \eigenvector_{z, 1}, \dots, \eigenvector_{z, r-1}$ are:
\begin{equation*}
\eigenvector_{z, j} = \eigenvector^{A}_{z \bmod \QCCam , j}, \quad \forall z \in C^{(r)}, j \in \interval{r}.
\end{equation*}
\end{lemma}
\begin{proof}
The polynomial $\prod_{h=0}^{\QCCl-1} \gen[][h][h]$ has a zero $\gamma^z$ of multiplicity $r$ if and only if the polynomial $\prod_{h=0}^{\QCCl-1} \genarg{A}[h][h]{X^{\intb \QCCbn}}$ has a zero $\gamma^z$ of multiplicity $r$ (i.e., exactly $r$ polynomials $\genarg{A}[h][h]{X}$ have a zero at $\alpha^z$) and the polynomial $\genarg{B}{X^{\inta \QCCan}}$ is nonzero if evaluated at $\gamma^z$ (see~\cite[Thm. 3]{lin_further_1970}). The cardinality $|C^{(r)}|$ follows.

With $\gamma= \alpha \beta$ we obtain for~\eqref{eq_QCTimeCyclicDiag} and \eqref{eq_QCTimeCyclicNonDiag}, that
\begin{align*}
v_h(\gamma^z) \genarg{A}[h][h^{\prime}]{(\alpha \beta)^{z \intb \QCCbn}} \genarg{B}{(\alpha \beta)^{z \inta \QCCan}} & = v_h(\gamma^z) \genarg{A}[h][h^{\prime}]{\alpha^{z \intb \QCCbn}} \genarg{B}{\beta^{z \inta \QCCan}} \\
  & = v_h(\gamma^z) \genarg{A}[h][h^{\prime}]{\alpha^{z}} \genarg{B}{\beta^{z}}, \qquad \qquad \forall h \in \interval{\QCCl}, h^{\prime} \in \interval{h,\QCCl}.
\end{align*}
This allows us to rewrite:
\begin{equation*}
\bar{\mathbf{G}}(\gamma^{z}) = \diag \left( v_0(\gamma^z), v_1(\gamma^z), \dots, v_{\QCCl-1}(\gamma^z) \right) \genmat[A][\alpha^{z}] \genarg{B}{\beta^{z}}.
\end{equation*}
The right kernel of $\genmat[A][\alpha^{z}]$ is therefore contained in the right kernel of $\bar{\mathbf{G}}(\gamma^{z})$. Since these two kernels have the same cardinalities, it follows that they must be equal. 
\end{proof}
\begin{example}[Eigenvalues of a $2$-Quasi-Cyclic Product Code] \label{ex_BinaryQCCodeZeros}
Let the two codes $\QCCa$ and $\QCCb$ with generator matrix $\genmat[A] \in \Fx{2}^{2 \times 2}$ and generator polynomial $\gen[B] \in \Fx{2}$ be as in Example~\ref{ex_BinaryQCCode}.
Let $\xi$ denote a primitive element in $\F{2^{12}} \cong \Fx{2}/(X^{12} + X^7 + X^6 + X^5 + X^3 + X + 1)$, $\alpha=\xi^{195}$ be a $21^{\text{st}}$ root of unity, $\beta = \xi^{819}$ a \ith{5} root of unity and $\gamma = \alpha \beta = \xi^{1014}$ a \ith{105} root of unity in $\F{2^{12}}$. 

Clearly, all eigenvalues $\eigenvalue{i}^A = \alpha^{i}$ for all $i \in A^{(2)} =  \coset{1} \cup \coset{3} \cup \coset{7} = \{1,2,3,4,6,7,8,11,12,14,16\}$ are roots of $\gen[A][0][0]$ and $\gen[A][1][1]$, and have multiplicity two. The corresponding eigenvectors span the full space $\F{2^{12}}^2$. The defining set of $\QCCb$ is $B=\{0 \}$.
According to Lemma~\ref{lem_EigenvaluesMaxMulti}, we have:
\begin{align*}
C^{(2)} & = A^{(2)} \cup \shiftset{A^{(2)}}{21} \cup \shiftset{A^{(2)}}{42} \cup \shiftset{A^{(2)}}{63} \cup \shiftset{A^{(2)}}{84} \cup B \cup \shiftset{B}{5} \cup \shiftset{B}{10} \cup \cdots \cup \shiftset{B}{100} \\
 & =  \{ 1,2,3,4,6,7,8,11,12,14,16\} \cup \{22,23,24,25,27,28,29,32,33,35,37 \} \cup \{43,\dots,58  \} \cup \{64,\dots, 79 \} \\ 
& \qquad  \cup \{85,\dots,100  \} \cup \{ 0 \} \cup \{ 5  \} \cup \cdots \cup \{ 100 \} \\
& = \{0,1,2,3,4,5,6,7,8,10,11,12,13,15,16,\dots,102 \}
\end{align*}
as the set of exponents of all eigenvalues $\eigenvalue{i} = \gamma^i, \forall i \in C^{(2)}$ of (maximal) multiplicity two. We have
\begin{equation*}
|C^{(2)}| = 21 (5-4) + 11 \cdot 4 = 65.
\end{equation*}
The eigenvalues $\eigenvalue{i}^A = \alpha^{i}$ for all $i \in A^{(1)} = \coset{9} = \{9,15,18 \}$ have multiplicity one and the eigenvalues $\alpha^{i}$ for all $i \in A^{(0)} = \coset{0}  \cup \coset{5}  = \{0, 5, 10, 13, 17, 19, 20 \}$ have multiplicity zero. The two sets 
\begin{align*}
C^{(1)} & = \big\{ 9, 18, 36, 39, 51, 57, 72, 78, 81, 93, 99, 102 \big\} \; \text{and}\\
C^{(0)} & = \big\{ 13, 17, 19, 21, 26, 31, 34, 38, 41, 42, 47, 52, 59, 61, 62, 63, 68, 73, 76, 82, 83, 84, 89, 94, 97, 101, 103, 104  \big\}
\end{align*}
contain the exponents of the eigenvalues of multiplicity one and zero respectively (according to Lemma~\ref{lem_EigenvaluesSmallerThanMax}).
We obtain:
\begin{align*}
|C^{(1)}| & = 3 \cdot 4 = 12 \quad \text{and} \\
|C^{(0)}| & = 7 \cdot 4 = 28.
\end{align*}
We explicitly calculate an eigenvector of multiplicity one. For $\eigenvalue{9}^A = \alpha^9$ we get:
\begin{equation*}
\genmat[A][\alpha^9] = 
\begin{pmatrix}
\minpoly{\alpha^1}[\alpha^9] \minpoly{\alpha^3}[\alpha^9] \minpoly{\alpha^7}[\alpha^9] & \minpoly{\alpha^1}[\alpha^9] \minpoly{\alpha^3}[\alpha^9] \minpoly{\alpha^7}[\alpha^9] \cdot (1+\alpha^9+\alpha^{18}) \\
0 & 0
\end{pmatrix}
\end{equation*}
and a corresponding eigenvector is 
\begin{equation} \label{ex_SpeficifEigenvector}
\eigenvector^{A}_{9,0} = \left( 1 \hspace{.4cm}   \xi^{11} + \xi^{10} + \xi^8 + \xi^7 + \xi^6 + \xi^2 + \xi \right)^T,
\end{equation}
and according to Lemma~\ref{lem_EigenvaluesSmallerThanMax}, the $2$-quasi-cyclic product code $\QCCa \otimes \QCCb$ has eigenvectors $\eigenvector_{9,0} = \eigenvector_{51,0} = \eigenvector_{72,0} = \eigenvector_{93,0} = \eigenvector^{A}_{9,0}$. The eigenvalue $\eigenvalue{30}$ is of multiplicity two, and the corresponding eigenvectors $\eigenvector_{30,0}$ and $\eigenvector_{30,1}$ span the space $\F{2^{12}}^2$.
\end{example}

\subsection{Bounding the Minimum Hamming Distance} \label{subsec_BoundingDistance}
In the following, we recall the BCH-like lower bound on the minimum Hamming distance of a quasi-cyclic code based on the spectral analysis of Semenov and Trifonov~\cite{semenov_spectral_2012}, because we use this fact subsequently.
\begin{theorem}[BCH-like Bound on the Minimum Hamming Distance of a Quasi-Cyclic Code~{\cite[Thm. 2]{semenov_spectral_2012}}] \label{theo_SemenovTrifonovBound}
Let $\QCC$ be an $\LIN{\QCCl \cdot \QCCm}{\QCCk}{\QCCd}{q}$ $\QCCl$-quasi-cyclic code, $\alpha$ an element of order $\QCCm$ in $\F{q^{\QCCs}}$, and let the set 
\begin{equation*}
D \defeq \big\{ \HTconst, \HTconst + \HTmult, \HTconst + 2\HTmult, \dots, \HTconst + (\lenseq-2)\HTmult \big\}\end{equation*} 
for some integers $\HTconst \geq 0, \HTmult >0, \lenseq > 2 $ with $\gcd(\HTmult, \QCCm)=1$ be given. Let the eigenvalues $\lambda_i = \alpha^i, \forall i \in D$ and their corresponding eigenspaces $\eigenspace[i]$ for all $i \in D$ be given. Define the intersection of eigenspaces $\eigenspace \defeq \cap_{i \in D} \eigenspace[i]$ and let $\eigencode{\eigenspace}$ be the corresponding eigencode as in Definition~\ref{def_eigencode} with distance $\ECd$. If
\begin{equation*}
\sum_{i=0}^{\infty} \mathbf{c}(\alpha^{\HTconst + i \HTmult}) \circ \eigenvector X^i \equiv 0 \mod X^{\lenseq-1}
\end{equation*}
holds for all $\mathbf{c}(X) = (c_0(X) \ c_1(X) \ \cdots \ c_{\QCCl-1}(X) ) \in \QCC$ and for all $\eigenvector = (\eigenvector[0] \ \eigenvector[1] \ \cdots  \ \eigenvector[\QCCl-1]) \in \eigenspace$, then, $\QCCd \geq \min(\lenseq, \ECd)$.
\end{theorem}
\begin{proof}
See proof of \cite[Thm. 2]{semenov_spectral_2012} or proof of~\cite[Thm. 1 for $\noseq = 0$]{zeh_decoding_2014}.
\end{proof}
Similar to the embedding of a given cyclic code $\QCCa$ into a cyclic product codes $\QCCa \otimes \QCCb$ as in~\cite[Thm. 4]{zeh_generalizing_2013}, we propose in the following theorem a new lower bound on the minimum Hamming  distance of a given $\QCCl$-quasi-cyclic code $\QCCa$ by embedding it into an $\QCCl$-quasi-cyclic product code $\QCCa \otimes \QCCb$.
\begin{theorem}[Generalized Semenov--Trifonov Bound] \label{theo_GeneralizedBCHBound}
Let $\QCCa$ be an $\LIN{\QCCl \cdot \QCCam}{\QCCak}{\QCCad}{q}$ $\QCCl$-quasi-cyclic code, $\alpha$ an element of order $\QCCam$ in $\F{q^{\QCCas}}$, $\QCCb$ an $\LIN{\QCCbn = \QCCbm}{\QCCbk}{\QCCbd}{q}$ cyclic code, and $\beta$ an element of order $\QCCbm$ in $\F{q^{\QCCbs}}$. Furthermore, let $\gcd(\QCCam, \QCCbm) = 1$. 

Let the integers $\HTconsta \geq 0, \HTconstb \geq 0 , \HTmulta > 0, \HTmultb > 0, \lenseq >2$ with $\gcd(\HTmulta, \QCCam)=1$, and $\gcd(\HTmultb, \QCCbm)=1$ be given, such that:
\begin{equation} \label{eq_ZerosGeneralizedSTBound}
\sum_{i=0}^{\infty} \left( \mathbf{a}(\alpha^{\HTconsta + i \HTmulta}) \cdot b(\beta^{\HTconstb + i \HTmultb}) \right) \circ \eigenvector X^i \equiv 0 \mod X^{\lenseq-1}
\end{equation}
holds for all $\mathbf{a}(X) = (a_0(X) \ a_1(X) \ \cdots \ a_{\QCCl-1}(X) ) \in \QCCa$, $b(X) \in \QCCb$, and for all $\eigenvector = (\eigenvector[0] \ \eigenvector[1] \ \cdots  \ \eigenvector[\QCCl-1]) \in \F{q^{\QCCas}}^{\QCCl} $ in the intersection of the eigenspaces
\begin{equation} \label{eq_EigenspaceGeneralizedBound}
\eigenspace \defeq \cap_{j \in D} \eigenspace[j],
\end{equation}
where 
\begin{equation} \label{eq_SetofEigenvaluesGeneralizedBound}
D = \Big\{ \HTconsta + i \HTmulta \, | \; b(\beta^{\HTconstb + i \HTmultb}) \neq 0, \quad \forall i \in \interval{\lenseq-1} \Big\}.
\end{equation}
Let the distance of the eigencode $\eigencode{\eigenspace}$ be $\ECd$. 
Then:
\begin{equation} \label{eq_BoundOnTheDistance}
\QCCad \geq \NewBCHBound \defeq \left \lceil \frac{\min(\lenseq, \ECd)}{\QCCbd}   \right \rceil.
\end{equation}
\end{theorem}
\begin{proof}
Let $\gamma=\alpha \beta$. The sequence $\mathbf{a}(\alpha^{\HTconsta}) b(\beta^{\HTconstb})\circ \eigenvector, \mathbf{a}(\alpha^{\HTconsta + \HTmulta}) b(\beta^{\HTconstb + \HTmultb})\circ \eigenvector, \dots, \mathbf{a}(\alpha^{\HTconsta + (\lenseq-2) \HTmulta}) b(\beta^{\HTconstb + (\lenseq-2) \HTmultb})\circ \eigenvector$ of $\lenseq-1$ zeros corresponds to $\lenseq-1$ zeros $\mathbf{c}(\gamma^{\HTconst + i \HTmult}) \circ \eigenvector$ for all $i \in \interval{\lenseq-1}$, where $\mathbf{c}(X)$ is a codeword of the $\QCCl$-quasi-cyclic product code $\QCCa \otimes \QCCb$ (see \cite[Prop. 1]{zeh_generalizing_2013} for the values of $\HTconst$ and $\HTmult$). Hence, the minimum Hamming distance $\QCCad \QCCbd$ of the product code $\QCCa \otimes  \QCCb$ is at least $\min(\lenseq, \ECd)$ due to the lower bound of Thm.~\ref{theo_SemenovTrifonovBound} on the minimum Hamming distance of $\QCCa \otimes \QCCb$.
\end{proof}
Note that the zeros of the cyclic code $\QCCb$ correspond to eigenvalues of multiplicity $\QCCl$ of the $\QCCl$-quasi-cyclic product code $\QCCa \otimes \QCCb$ and do not influence the intersection of the eigenspaces.
\begin{example}[Bound via Quasi-Cyclic Product Code] \label{ex_BinaryQCCodeBCHBound}
We consider the $\LIN{2 \cdot 21}{17}{8}{2}$ $2$-quasi-cyclic code $\QCCa$ and the $\LIN{5}{4}{2}{2}$ cyclic single-parity check code $\QCCb$ as in Example~\ref{ex_BinaryQCCodeZeros}.
For $\HTconsta = 0, \HTmulta = 1, \HTconstb = 0, \HTmultb = 1$, and the set 
\begin{equation*}
D = \{1,2,3,4,6,7,8,9,11,12\},
\end{equation*}
Thm.~\ref{theo_GeneralizedBCHBound} holds for $\lenseq=14$ with $\ECd = \infty$ and therefore $\QCCad \geq \lceil 14/2  \rceil = 7$.
More explicitly, the sequence of length $\lenseq-1 = 13$ is:
\begin{align*}
\mathbf{a}(\alpha^{0}) b(\beta^{0}), \mathbf{a}(\alpha^{1}) b(\beta^{1}), \dots, \mathbf{a}(\alpha^{9}) b(\beta^{4}), \dots, \mathbf{a}(\alpha^{12}) b(\beta^{2}). 
\end{align*}
The defining set $B = \{0\}$ of the associated cyclic code $\QCCb$ of length $\QCCbn=5$ fill the ``gaps'' at position $0,5,10$. The eigenvalues of $\QCCa \otimes \QCCb$ that correspond to the product $\mathbf{a}(\alpha^{i}) b(\beta^{i}),  \forall i \in \interval{13}$ have multiplicity two, except the one that relates to $\mathbf{a}(\alpha^{9}) b(\beta^{4})$. 
The eigenspace $\eigenspace[9]$ of $\QCCa$ has (geometric) multiplicity one and is generated by the eigenvector $\eigenvector_{9,0}^{A}$ as in~\eqref{ex_SpeficifEigenvector}. The two entries of $\eigenvector_{9,0}^{A}$ are linearly independent over $\F{2}$ and therefore $\ECd = \infty$.

The BCH-like bound as in Thm.~\ref{theo_SemenovTrifonovBound} for $\QCCa$ states that the minimum Hamming distance of $\QCCa$ is at least  five. The Hartmann--Tzeng-like~\cite{hartmann_generalizations_1972} lower bound as shown in \cite{zeh_decoding_2014} gives six. Therefore Thm.~\ref{theo_GeneralizedBCHBound} gives an improvement over these two bounds in this case.
\end{example}

\section{Syndrome-Based Phased Burst Error Correction up to the New Bound} \label{sec_Decoding}

Let $\QCCa$ be an $\LIN{\QCCl \cdot \QCCam}{\QCCak}{\QCCad}{q}$ $\QCCl$-quasi-cyclic code and let $\QCCb$ be an $\LIN{\QCCbm}{\QCCbk}{\QCCbd}{q}$ cyclic code as in Thm.~\ref{theo_GeneralizedBCHBound}. Moreover, we assume throughout this section that there exists an eigenvector $(\eigenvector[0] \ \eigenvector[1] \ \cdots  \ \eigenvector[\QCCl-1]) \in \F{q^{\QCCas}}^{\QCCl}$ in the intersection of the eigenspaces as in Thm.~\ref{theo_GeneralizedBCHBound} with entries $\eigenvector[0], \eigenvector[1],\dots,\eigenvector[\QCCl-1]$ that are linearly independent over $\F{q}$ and therefore $\QCCad \geq \NewBCHBound = \left \lceil \lenseq / \QCCbd   \right \rceil$.

Similar to our approach for cyclic codes~\cite{zeh_decoding_2012, zeh_new_2014}, we develop a syndrome-based decoding algorithm for a given $\QCCl$-quasi-cyclic code $\QCCa$, which guarantees to correct up to $\lfloor (\NewBCHBound-1)/2 \rfloor$ $\QCCl$-phased burst errors in $\F{q}$. We define syndromes, derive a key equation, and describe the algorithm with guaranteed burst error decoding radius.

After transmitting a codeword $(a_0(X) \ a_1(X) \ \cdots \ a_{\QCCl-1}(X)) \in \QCCa$, let the received word be:
\begin{align*}
\mathbf{r}(X) & = \big( r_0(X) \ r_1(X) \ \cdots \ r_{\QCCl-1}(X) \big) = \big( a_0(X) + e_0(X) \ a_1(X) + e_1(X) \ \cdots \ a_{\QCCl-1}(X) + e_{\QCCl-1}(X) \big),
\end{align*}
where
\begin{equation*} \label{eq_errorword}
e_j(X) = \sum_{i \in \errorsup[j]} e_{j,i} X^i, \quad j \in \interval{\QCCl},
\end{equation*}
are $\QCCl$ error polynomials in $\Fx{q}$ of weight $\noerrors[j] \defeq |\errorsup[j]|$ and degree less than $\QCCam$. An $\QCCl$-phased \textit{burst error} at position $i$ consists of at least one \nonzero{} entry $e_{0,i}, e_{1,i}, \dots, e_{\QCCl-1,i} \in \F{q}$. The cardinality of the set: 
\begin{equation*} \label{eq_ModError}
\errorsup \defeq \bigcup_{j=0}^{\QCCl-1} \errorsup[j] \subseteq \interval{\QCCam}. 
\end{equation*} 
of $\QCCl$-phased burst errors is denoted by $\noerrors \defeq  |\errorsup|$.
\printalgo{\renewcommand{\algorithmcfname}{Algo}
\caption{\textsc{Decoding an $\LIN{\QCCl \cdot \QCCam}{\QCCak}{\QCCad \geq \NewBCHBound}{q}$ $\QCCl$-Quasi-Cyclic Code up to $\lfloor (\NewBCHBound-1)/2 \rfloor $ $\QCCl$-phased Burst Errors}}
\label{algo_DecodingAlgorithm} 
\DontPrintSemicolon 
\SetAlgoVlined
\LinesNumbered
\newcommand\mycommfont[1]{\footnotesize\ttfamily{#1}}
\SetCommentSty{mycommfont}
\SetKwInput{KwFake}{}
\SetKwInput{KwIn}{Input}
\SetKwInput{KwOut}{Output}
\SetKwInput{KwPre}{Preprocess}
\BlankLine
\KwIn{\textcolor{white}{\textbf{O}}Parameters $\QCCl,\QCCam,\QCCak,\QCCad, q$ of $\QCCa$ and $\alpha \in \F{q^{\QCCas}}$\\
\textcolor{white}{\textbf{Output}: }Parameters $\CYCbn,\CYCbk,\CYCbd$ of $\QCCb$, and a codeword $b(X) =  \sum_{i \in \mathcal{W}} b_i X^i \in \QCCb $ with $|\mathcal{W}| = \QCCbd$, and $\beta \in \F{q^{\QCCbs}}$\\
\textcolor{white}{\textbf{Output}: }Integers $\HTconsta \geq 0, \HTconstb \geq 0, \lenseq > 2$, and $ \HTmulta > 0, \HTmultb > 0$ with $\gcd(\HTmulta,\QCCam)=1$, and $\gcd(\HTmultb,\QCCbm)=1$ \\
\textcolor{white}{\textbf{Output}: }\hspace{0.2cm} as in Thm.~\ref{theo_GeneralizedBCHBound}, and an eigenvector $(\eigenvector[0] \ \eigenvector[1] \ \cdots \ \eigenvector[\QCCl-1]) \in \F{q^{\QCCas}}^{\QCCl}$ with $\F{q}$-linearly independent entries\\ 
\textcolor{white}{\textbf{Output}: }Received word $\mathbf{r}(X) = (r_0(X) \ r_1(X) \ \cdots \ r_{\QCCl-1}(X)) \in \Fx{q}^{\QCCl}$
}
\KwOut{Estimated codeword $ \mathbf{a}(X)=(a_0(X) \ a_1(X) \ \cdots \ a_{\QCCl-1}(X) ) \in \QCCa $  or \textsc{Decoding Failure}}
\KwPre{\\
\hspace{1em}\textbf{for} all $i \in \interval{\QCCam}$: calculate $\gamma_i = \beta^{-j\HTmultb} \alpha^{-i\HTmulta} $, where $j \in \mathcal{W}$
}
\BlankLine 
\BlankLine 
Calculate syndrome polynomial $S(X)$ as in~\eqref{eq_DefSyndromes} \nllabel{algo_SyndCalc} 
\BlankLine
Solving Key Equation $(\Lambda(X), \Omega(X))$ = \texttt{EEA$\big( S(X), X^{\lenseq-1} \big)$} \nllabel{algo_KeyEquation} \tcc*[r]{Extended Euclidean Algorithm} 
\BlankLine
Find all $i \in \interval{\QCCam}$, where $\Lambda(\gamma_i)=0$, $ \rightarrow {\mathcal E}=\lbrace i_0,i_1,\dots,i_{\noerrors-1}\rbrace$ \nllabel{algo_RootFinding} \tcc*[r]{Chien-like Root-Finding} 
\BlankLine
\If{$ \noerrors \CYCbd < \deg \Lambda(X)$}
{Declare \textsc{Decoding Failure}}
\Else{
Determine $e_{0,i_j}, e_{1, i_j}, \dots, e_{\QCCl-1, i_j} \in \F{q}, \quad \forall i_j \in \mathcal E$ \nllabel{algo_SmallErrorValues} \tcc*[r]{Error-Evaluation as in~\cite[Prop. 4]{zeh_new_2014}}
\BlankLine
$e_j(X) \leftarrow \sum_{i \in \mathcal{E}_j} e_{j,i} X^{i}, \quad \forall j \in \interval{\QCCl}$ \;
\BlankLine
$a_j(X) \leftarrow r_j(X) - e_j(X), \quad \forall j \in \interval{\QCCl}$\;}
}
Algorithm~\ref{algo_DecodingAlgorithm} is the decoding procedure for a given $\QCCl$-quasi-cyclic code $\QCCa$ that is guaranteed to decode up to
\begin{equation*} 
\tau \leq \left \lfloor \frac{\NewBCHBound-1}{2} \right \rfloor.
\end{equation*}
$\QCCl$-phased burst errors. Let 
\begin{equation*}
b(X) = \sum_{j \in \mathcal{W}} b_j X^j
\end{equation*}
be a codeword of weight $|\mathcal{W}| =  \QCCbd$ of the associated $\LIN{\QCCbn}{\QCCbk}{\QCCbd}{q}$ cyclic code $\QCCb$ with zeros in $\F{q^{\QCCbs}}$. Let $\HTconsta, \HTconstb, \HTmulta, \HTmultb, \lenseq$, and an eigenvector $\eigenvector = (\eigenvector[0] \ \eigenvector[1] \ \cdots \ \eigenvector[\QCCl-1]) \in \eigenspace \subseteq \F{q^{\QCCas}}^{\QCCl}$ be given as in Thm.~\ref{theo_GeneralizedBCHBound}, where the entries $\eigenvector[0], \eigenvector[1], \dots,\eigenvector[\QCCl-1]$ are linearly independent over $\F{q}$. Let $\QCCs = \lcm(\QCCas, \QCCbs)$. Define the following syndrome polynomial in $\Fx{q^{\QCCs}}$: 
\begin{align}
S(X) & \defequiv  \sum_{i=0}^{\infty} \left( \sum_{j=0}^{\QCCl-1} r_j(\alpha^{\HTconsta +i \HTmulta}) b(\beta^{ \HTconstb + i \HTmultb})  \eigenvector[j] \right) X^i \mod X^{\lenseq-1} \nonumber \\
& \defeqspace  \sum_{i=0}^{\lenseq-2} \left( \sum_{j=0}^{\QCCl-1} r_j(\alpha^{\HTconsta +i \HTmulta})  b(\beta^{ \HTconstb + i \HTmultb})  \eigenvector[j] \right) X^i.  \label{eq_DefSyndromes}
\end{align}
From Thm.~\ref{theo_GeneralizedBCHBound} it follows that the syndrome polynomial $S(X)$ as defined in~\eqref{eq_DefSyndromes} is independent of a codeword in $\QCCa$ and therefore the expression of~\eqref{eq_DefSyndromes} for the syndrome polynomial can be rewritten as: 
 \begin{equation} \label{eq_SyndError}
S(X) =  \sum_{i=0}^{\lenseq-2} \left( \sum_{j=0}^{\QCCl-1} e_j(\alpha^{\HTconsta +i \HTmulta})  b(\beta^{ \HTconstb + i \HTmultb})  \eigenvector[j] \right) X^i.
\end{equation}
Define an error-locator polynomial in $\Fx{q^{\QCCs}}$:
\begin{equation} \label{eq_ELP}
\Lambda(X) \defeq \sum_{i=0}^{\QCCbd \noerrors} \Lambda_i X^i \defeq \prod_{i \in \errorsup} \prod_{j \in \mathcal{W}} \left(1-X\alpha^{\HTmulta i}\beta^{\HTmultb j} \right),
\end{equation}
which depends on the position of the burst error and on the nonzero lowest-weight codeword of the associated cyclic code $\QCCb$.

For some $j \in \mathcal{W}$, define $\QCCam$ elements in $\F{q^\QCCs}$ as:
\begin{equation} \label{eq_RootOfELP}
\gamma_i \defeq \beta^{-j \HTmultb} \alpha^{-i \HTmulta}, \quad \forall i \in \interval{\QCCam}.
\end{equation}
We pre-calculate the $\QCCam$ values as in~\eqref{eq_RootOfELP} to identify the roots of a given error-locator polynomial $\Lambda(X)$ as in~\eqref{eq_ELP} (see Line~\ref{algo_RootFinding} of Algorithm~\ref{algo_DecodingAlgorithm}).

Combining the syndrome definition as in~\eqref{eq_SyndError} and definition of the error-locator polynomial as in~\eqref{eq_ELP} gives, like in the classical case of cyclic codes, a \textit{Key Equation} of the following form:
\begin{equation} \label{eq_KeyEquation}
 \Lambda(X) \cdot S(X) \equiv \Omega(X) \mod X^{\lenseq-1},
 \end{equation}
where the degree of the so-called \textit{error-evaluator} polynomial $\Omega(X)$ is smaller than $\QCCbd \noerrors$.

Solving the Key Equation~\eqref{eq_KeyEquation} can be realized by shift-register synthesis or the Extended Euclidean Algorithm (EEA). We use the EEA in Line~\ref{algo_KeyEquation} of Algorithm~\ref{algo_DecodingAlgorithm} that returns the error-locator polynomial $\Lambda(X)$ and the error-evaluator polynomial $\Omega(X)$ given the syndrome polynomial $S(X)$ and the monomial $X^{\lenseq-1}$. Determining the error-values (Line~\ref{algo_SmallErrorValues} in Algorithm~\ref{algo_DecodingAlgorithm}) is straightforward (see, e.g.,~\cite[Prop. 4]{zeh_new_2014}).

Our syndrome-based decoding approach can be easily extended to the case of a $\kappa$-interleaved code, i.e., a code that consists of $\kappa$ vertically arranged $\QCCl$-quasi-cyclic codes. If errors occur, in addition to the $\QCCl$-phased arrangement within each vector in $\F{q}^{\QCCl \QCCam}$, as $\kappa$-phased burst errors in the interleaved code, we obtain overall $\kappa \QCCl$-phased burst errors in $\F{q}$. Then, the $\kappa$ Key Equations as in~\eqref{eq_KeyEquation} have a common error-locator polynomial, which allows collaborative decoding up to 
\begin{equation*}
\left \lfloor \frac{\kappa}{\kappa+1} (\NewBCHBound-1) \right \rfloor
\end{equation*}
$\kappa \QCCl$-phased burst errors with high probability (analyzed, e.g., in~\cite{krachkovsky_decoding_1997, krachkovsky_decoding_1998, schmidt_collaborative_2009}). In the case of $\kappa =2$, this gives for the binary 2-quasi-cyclic code of Example~\ref{ex_BinaryQCCodeBCHBound} a collaborative decoding radius of 4.

\section{Conclusion and Outlook} \label{sec_conclusion}
We have derived an unreduced basis of an $\QCCal \QCCbl$-quasi-cyclic product code in terms of the given generator matrices in RGB/POT form of the $\QCCal$-quasi-cyclic row-code and the $\QCCbl$-quasi-cyclic column-code. For two special cases, the generator matrix in Pre-RGB/POT form of the $\QCCal \QCCbl$-quasi-cyclic product code was derived. The general expression for the reduced basis was conjectured.

Based on spectral analysis, a technique for bounding the minimum Hamming distance of a given $\QCCl$-quasi-cyclic code via embedding it into an $\QCCl$-quasi-cyclic product code was outlined, which outperforms existing known bounds in many cases. We have proposed an algebraic decoding algorithm with guaranteed $\QCCl$-phased burst error correction radius. 

Beside the proof of Conjecture~\ref{conj_ProductCodeQCCQCC}, the investigation of concatenated quasi-cyclic codes (see~\cite{blokh_coding_1974} and~\cite{jensen_cyclic_1992}) is open future work. 
 Furthermore, an extension of the embedding technique to an interpolation-based list decoding algorithm (see \cite{zeh_improved_2015}) seems possible.

\printbibliography
\end{document}